\documentclass[a4paper,twoside,12pt]{amsart}
\usepackage{amsmath,amsthm}
\usepackage{amssymb}
\usepackage{bbm}
\usepackage{graphics}
\usepackage[all,cmtip]{xy}
\numberwithin{equation}{section}%
\newcommand{\R}{\mathbbm{R}}
\newcommand{\N}{\mathbbm{N}}
\newcommand{\C}{\mathbbm{C}}
\newtheorem{Satz}{Theorem}[section]
\newtheorem{Def}[Satz]{Definition}
\newtheorem{Lemma}[Satz]{Lemma}
\newtheorem{Kor}[Satz]{Corollary}

\newtheorem{Hyp}{Hypothesis}
\newcommand{\disp}{\eta}
\newcommand{\dispo}{\eta_0}
\newcommand{\tensor}{\otimes}

\newcommand{\Fock}{\mathcal{F}}
\newcommand{\Tr}{\operatorname{Tr}}

\newcommand{\cl}{\operatorname{cl}}
\newcommand{\dom}{\operatorname{dom}}

\newcommand{\ovl}{\overline}
\newcommand{\unl}{\underline}
\newcommand{\Hel}{\mathcal{H}_{el}}

\newcommand{\Lel}{\mathcal{L}_{el}}
\newcommand{\Hael}{H_{el}}
\newcommand{\Hosc}{H_{osc}}

\newcommand{\tauel}{\tau^{el}}

\newcommand{\We}{W}
\newcommand{\const}{\mathfrak{c}}
\newcommand{\Af}{\mathfrak{A}_f}

\newcommand{\omf}{ \omega_f}

\newcommand{\Hf}{ \mathcal{H}_{f} }
\newcommand{\Haf}{ H_f }
\newcommand{\Kf}{ \mathcal{K}_f }
\newcommand{\Omf}{\Omega_{f}}
\newcommand{\Lf}{ \mathcal{L}_f}
\newcommand{\LQ}{ \mathcal{L}_Q}

\newcommand{\Mf}{\mathfrak{M}_f}
\newcommand{\Mfa}{{\mathfrak{M}_f^a}}
\newcommand{\pig}{\pi_0}
\newcommand{\Ag}{\mathfrak{A}}
\newcommand{\Qg}{Q}
\newcommand{\Jg}{\mathcal{J}}
\newcommand{\Lg}{\mathcal{L}_h}
\newcommand{\Lo}{\mathcal{L}_0}
\newcommand{\Hg}{H_{h}}
\newcommand{\Ho}{H_0}
\newcommand{\Kg}{\mathcal{K}}
\newcommand{\Hig}{\mathcal{H}}

\newcommand{\omo}{\omega_0}
\newcommand{\Mg}{\mathfrak{M}}

\newcommand{\linhull}{\operatorname{LH}}
\newcommand{\Hil}{\mathcal{H}}
\newcommand{\h}{\mathfrak{h}}
\newcommand{\g}{\mathfrak{g}}
\newcommand{\f}{\mathfrak{f}}
\newcommand{\one}{\mathbbm{1}}

\newcommand{\slim}{\operatorname{s-lim}}
\renewcommand{\leq}{\leqslant}
\renewcommand{\le}{\leqslant}
\renewcommand{\geq}{\geqslant}
\renewcommand{\ge}{\geqslant}
\renewcommand{\Re}{\textrm{Re}\,}
\renewcommand{\Im}{\textrm{Im}\,}
\renewcommand{\imath}{\operatorname{i}}
\newcommand{\W}[1]{\mathcal{W}[{#1}]}
\newcommand{\taug}{\tau}

\newcommand{\tauf}{ \tau^f}
\newcommand{\Omo}{\Omega_0}
\newcommand{\Omg}{\Omega}
\newcommand{\Omp}{\Omega_Q}
\newcommand{\omg}{\omega}
\newcommand{\omp}{\omega_Q}
\newcommand{\Aobs}{\tilde{\mathfrak{A}}}
\newcommand{\Ao}{\mathfrak{A}_0}
\newcommand{\tev}[2]{\alpha_{#1}({#2})}
\newcommand{\tevabb}[1]{\alpha_{#1}}
\newcommand{\piaw}{\pi_{AW}}
\newcommand{\tevf}[2]{\alpha^f_{#1}({#2})}
\newcommand{\taup}[2]{\alpha^Q_{#1}({#2})}

\newcommand{\wlim}{\operatorname{w-lim}}
\begin{document}
\title[Return to Equilibrium]{Return to Equilibrium for an Anharmonic Oscillator
coupled to a Heat Bath}
\author{Martin K\"onenberg}
\address{
Fakult\"at f\"ur Mathematik und Informatik\\
FernUniversit\"at Hagen\\
L\"utzowstra{\ss}e 125\\
D-58084 Hagen, Germany.}
\email{martin.koenenberg@fernuni-hagen.de}
\date{\today}
\begin{abstract}
We study a $C^*$-dynamical system describing a particle
coupled to an infinitely extended heat bath at positive temperature. 
For small coupling constant we prove return to equilibrium exponentially fast in time.
The novelty in this context is to model the
particle by a harmonic or anharmonic oscillator, respectively. 
The proof is based on explicit formulas for
the time evolution of Weyl operators in the harmonic oscillator
case. In the anharmonic oscillator case, a
Dyson's expansion for the dynamics is essential. Moreover, 
we show in the harmonic oscillator case, that $\R$ 
is the absolute continuous spectrum of the Standard Liouvillean 
and that zero is a unique eigenvalue.
\end{abstract}
\maketitle
\section{ Introduction}
In this paper, we study an interacting system of
a single particle coupled to a heat bath, which 
is infinitely extended and near its thermal 
equilibrium at inverse temperature $\beta$.
The heat bath consists in infinitely many bosons with a
momentum density given by Planck's law for the
black body radiation. The particle is confined 
by an increasing potential, which inhibits 
an escape to infinity. In this
situation one expects that the interacting system is driven
to a joint equilibrium state at inverse temperature $\beta$
as time tends to infinity. This
behavior is called `return to equilibrium'. \\
\indent
The mathematical model is formulated in the framework 
of a quantum dynamical system. Here the observables are
modeled by a Weyl algebra and the automorphism group is
implemented by conjugating with a group of unitaries,
which is generated by the Hamiltonian of the interacting 
system. In this paper the particle Hamiltonian 
is either a harmonic or an anharmonic
oscillator and the interaction with the heat bath is given by a
dipole expression.\\
\indent In the last decade small systems coupled to
a heat bath were subject of extensive mathematical 
research. In 
\cite{JaksicPillet1996a} and \cite{JaksicPillet1996b} 
an approach is established that traces back the ergodic properties
of a certain $W^*$-dynamical system to the spectral 
characteristics of the so-called Standard Liouvillean. 
Moreover, the return to equilibrium occurs exponentially fast, 
the rate of decay is obtained by Fermi's golden rule for
the Standard Liouvillean. 
In \cite{BachFroehlichSigal2000} the spectral analysis of the Standard Liouvillean 
is studied by a renormalization group technique. The methods in
\cite{JaksicPillet1996a,JaksicPillet1996b} and \cite{BachFroehlichSigal2000}
require analytic form factor. The
assumptions on the singularity of the form factors at zero
are, however, less restrictive in \cite{BachFroehlichSigal2000}. \\
\indent In \cite{DerezinskiJaksic2001,DerezinskiJaksic2003} the
Liouvillean is studied by means of the limiting absorption principle
and the Feshbach method. In \cite{Merkli2001,FroehlichMerkliSigal2004} 
the positive commutator theorem and a viral theorem are applied to the Liouvillean.
In all papers mentioned above the Hilbert space representing the small system  is finite dimensional, 
and the interaction is small. In summary we call the strategy used by all authors
above \emph{Liouvillean approach}.\\
\indent The harmonic oscillator that interacts with the heat bath
by a dipole expression is a quadratic operator in annihilation and
creation operators, this allows to define a $\ast$-automorphism group $\taug$
on a Weyl algebra $\Ag$. Moreover, for any inverse temperature $\beta$
we can define a KMS state $\omg$ on the Weyl algebra $\Ag$, see Theorem \ref{Satz:Isom} below. 
This is a special property of the harmonic oscillator case. 
We benefit from very explicit results for the harmonic oscillator coupled
to a Boson field at temperature
zero, studied in \cite{Arai1981a,Arai1981c}, or in
\cite{Spohn1997} for a closely related model.\\
\indent We prove in Theorem \ref{Satz:RetEq} that the return to equilibrium 
occurs exponentially fast for a large class of states and observables for small interaction. 
The rate of decay is related to Fermi's golden rule for the Hamiltonian at temperature zero.
This is different from the approach using Liouvilleans, where the
rate of decay is deduced from Fermi's golden rule for the 
Standard Liouvillean, which is larger.\\
\indent We recall that the harmonic oscillator case is not
covered by the methods of the Liouvillean approach, yet. 
However, we can formulate another model  
analog to the Liouvillean approach. In this model we  
define a Standard Liouvillean for each
$\beta$, which generates the time evolution.
One can show existence of a KMS state on some $W^*$-algebra $\Mg$, 
see \cite{Koenenberg2009b}. 
However, we prove in Corollary \ref{Cor5.3} that
the spectrum of the Standard Liouvillean covers $\R$, it is absolutely 
continuous except in zero, and zero  a unique, non-degenerate eigenvalue. Moreover,
the model in the Liouvillean approach is an extension
of the quantum dynamical system formulated on the Weyl algebra $\Ag$ in
the following sense:\\
\indent The  Weyl algebra $\Ag$ can be embedded into $\Mg$ and the
dynamics generated by the Standard Liouvillean extends $\taug$, which
is defined on the Weyl algebra. This is stated in Theorem \ref{Thm5.1}.\\
\indent Moreover, we perturb the harmonic potential
by a potential $V$, which is the Fourier-transform of a complex measure.
The perturbed  Hamiltonian is called anharmonic oscillator.
The class of perturbations is adopted from \cite{Maassen1983},
where the Langevin equation is studied, see also \cite{Spohn1997} and \cite{FidaleoLiverani1999}.
The model is the following, we fix an inverse temperature $\beta$,
and construct a GNS representation associated with $\omg$.
The theory of KMS states ensures the existence of a KMS state $\omp$ 
for the perturbed dynamics in the Hilbert space of the GNS representation.
To obtain convergence to $\omp$ for large times we use Dyson's expansion 
for the perturbed dynamics. In Theorem \ref{Lem5.4}
and Corollary \ref{Kor1g} return to equilibrium is proved for the anharmonic
oscillator model with an exponential rate of decay for small 
coupling and small $V$. The strategy for the proof is based on an estimate,
we learned form \cite{Maassen1983}, it is based on the fact that
certain integrals that occur in Dyson's expansion decay exponential fast
in $t$. Recently, in \cite{BotvichMaassen2009} a combinatorial argument was
found, that relaxes this assumption on the decay. One could hope that
this result can be used to prove return to return to equilibrium for
the anharmonic oscillator, whether no analycity of the form factor is assumed, c.f.
Hypothesis \ref{Hyp1}.
\subsection{Organization of the Paper}
In the subsequent subsection we recall the formalism of second quantization.
We give here the definition of the Hamiltonian in the harmonic oscillator case
and formulate the Hypothesis on the form factors. Moreover, we
give definitions in the context of quantum dynamical
systems. 
In Section \ref{Sec4} we 
define the so-called analytic states and the analytic observables, for which 
the return to equilibrium is exponentially fast. We formulate
and prove Theorem \ref{Satz:Isom} and Theorem \ref{Satz:RetEq}.
Section \ref{Sec5} is devoted to the Liouvillean approach
in the harmonic oscillator case. In this section we prove
Theorem \ref{Thm5.1}, Theorem \ref{Thm5.2}
and Corollary \ref{Cor5.3}. 
The anharmonic oscillator 
coupled to a heat bath is studied in Section \ref{Sec6}, 
where we prove Theorem \ref{Lem5.4} and Corollary
\ref{Kor1g}.
The paper has two appendices: In the first we quote some definitions
and results given in \cite{Arai1981a,Arai1981c}, in the second we recall an estimate, which is
important for the proof of Theorem \ref{Lem5.4}, it was proved originally in \cite{Maassen1983}. 
\subsection{Notation and Definition}
The starting point is the state space $\Hil$,
which represents the coupled system of 
a particle and the bosons at temperature zero,
\begin{equation*}
\Hil =  L^2(\R) \tensor \mathcal{F}_b[\h].
\end{equation*}
The Hilbert space $L^2(\R)$ contains the states
for an isolated particle, which is for simplicity
assumed to be one-dimensional. The bosonic Fock space, 
\begin{equation*}
\mathcal{F}_b[\h]
 =  \C\Omega_{\h} 
\oplus \bigoplus_{n = 1}^\infty \mathcal{S}_n \bigotimes_{k=1}^n \h
\end{equation*}
is modeled over the one boson space $\h:= L^2(\R^3)$. 
Later on we will also use a Fock space over $\C\oplus \h$ or $\C\oplus \C\oplus \h \oplus \h$.
In this context $\Omega_{\h} $ is a fixed normed vector called the vacuum and
$\mathcal{S}_n$ is the projection onto the subspace of totally symmetric tensors. 
The vectors in $\mathcal{F}_b[\h]$ are sequences 
$\psi =  ( \psi_n )_{n = 0}^\infty$ such that
$\psi_n\in \mathcal{S}_n \bigotimes_{k=1}^n \h$ 
for $n \geq  1$ and $\psi_0\in \C\Omega_{\h}$.
The \emph{annihilation and creation operators} on $\mathcal{F}_b[\h]$ 
are denoted by $a(f),  a^*(g)$ for $f, g \in \h$. They satisfy the \emph{Canonical Commutator
Relations} (CCR),
\begin{gather}\label{Eq:CCR}
[  a(f),  a(g)  ] = 0,
\quad [  a^*(f),  a^*(g)  ] = 0,\\ \nonumber
[  a(f),  a^*(g) ]  = \langle  f | g \rangle_{\h}
\end{gather}
and $a(f) \Omega_{\h} = 0$. Furthermore, we define the \emph{field operators} by
\begin{equation*}
\Phi(f) =  \frac{1}{\sqrt{2}} \big(  a(f) + a^*(f)  \big).
\end{equation*}
$a(f), a^*(f)$ and $\Phi(f)$ are defined on the dense subspace of finite sequences
$(\psi_n)_{n=0}^\infty$ of $\mathcal{F}_b[\h]$. All three operators are
closable, we denote their closures by the same symbol. As suggested
in the notation $a^*(f)$ is the adjoint operator of $a(f)$. Moreover,
$\Phi(f)$ is self-adjoint, the exponential $\W{f}:= e^{\imath \Phi(f)}$
is called \emph{Weyl operator}. From \eqref{Eq:CCR} we deduce the CCR for Weyl operators
\begin{equation}\label{Eq:WeylCCR}
\W{f} \W{g}=e^{-\imath \Im\langle f | g\rangle_{\h}/2}  \W{f+g}.
\end{equation}
For a definition of the formalism of second quantization, $C^*$-algebras and related
topic we refer the reader to the textbooks \cite{BratteliRobinson1987,BratteliRobinson1996}.\\
\indent On the Schwartz space $\mathcal{S}(\R)\subset L^2(\R)$ we define by 
$(q \psi)(q) =  q\cdot \psi(q)$ the  position operator
and  by $(p \psi)(q)  = -\imath  \frac{d\psi}{d q}(q)$ the  momentum operator 
of the particle. Annihilation and creation operators in $L^2(\R)$ are given by
\begin{equation*}
A^* = \frac{1}{\sqrt{2}} ( q - \imath  p ),
\quad  A = \frac{1}{\sqrt{2}}( q + \imath p ).
\end{equation*}
They satisfy the CCR
\begin{equation*}
[ A, A ] = [ A^*, A^* ] = 0 ,\quad [ A, A^* ] = \one_{S(\R)},
\end{equation*}
in addition we have $A \Omega_{\C} = 0$ for $\Omega_{\C}(q) =  \pi^{-1/4}  e^{-q^2/2}$.
It is known that 
$
\cl_{L^2(\R)} \linhull \big\{  ( A^* )^n \Omega_{\C}\in L^2(\R) : n\in \N_0 \big\} = L^2(\R),
$
see for instance \cite{BratteliRobinson1996}.
The letter $\cl_{L^2(\R)}$ denotes the closure in the topology of $L^2(\R)$
and $\linhull$ denotes the linear hull. Thus we can identify $L^2(\R)$ with $\mathcal{F}_b[\C]$.
As a consequence we obtain $\Hil  \cong  \mathcal{F}_b[\C\oplus \h]$ 
and
\begin{gather*}
\Omega_{\C\oplus\h}  \cong  \Omega_{\C}\tensor \Omega_{\h},\quad
a(c\oplus h) \cong  \ovl{c}  A\oplus a(h),\\
a^*(c\oplus h) \cong  c  A^*\oplus a^*(h).
\end{gather*}
Moreover, we  field operators $\Phi(c\oplus f)$ and Weyl operators
$\W{c\oplus h}$ in $\Hig$ are given by
\begin{align} \label{Eq:DefFeldOp}
\Phi(c\oplus h)  &=  \frac{1}{\sqrt{2}}\big(a^*(c\oplus h) + a(c\oplus h)\big)\\
\W{c\oplus h}  &= \W{c}\tensor \W{h}.
\end{align}
For these operators the CCR in Equation \eqref{Eq:CCR} and
Equation \eqref{Eq:WeylCCR} are satisfied with 
$\langle\cdot | \cdot\rangle_{\h}$  replaced by 
$\langle\cdot | \cdot\rangle_{\C\oplus\h}$.
In this model the particle is confined by  the potential
$(1/2) x^2$, which inhibits an escape to infinity of the particle.
The harmonic oscillator is
\begin{equation*}
\Hosc  =  (1/2) (p^2 + q^2).
\end{equation*}
We remark, that a large class of potentials should ensure
return to equilibrium, but the choice of the harmonic potential
is essential for our analysis. The reason is, that
we can write the Hamiltonian in the formalism of second
quantization, i.e. $H_{osc}=A^*A+1/2$. This allows us to perform 
calculations explicitly.\\
\indent Throughout this paper the bosons are massless,
being modeled in the momentum space with the dispersion
relation $|k|$. The free Hamiltonian is defined by $\Haf\Omega_\h:=0$
and
\begin{equation}\label{DefHf}
(\Haf\psi)_n:=
\big(\sum_{j=1}^n \one \tensor\ldots\tensor 
         \underbrace{h_{ph}}_{j}\tensor\ldots\tensor \one\big) \psi_n,
\end{equation}
where $h_{ph}$ is the multiplication with $|k|$ in $\h$.
In the following we will not introduce an extra symbol
for multiplication operators, so $h_{ph}$ is just $|k|$.\\
\noindent The interaction operator is given by
\begin{equation*} 
H_I  = \lambda   q\cdot \Phi(|k|^{-1/2} \hat{\rho}),
\end{equation*}
where $\Phi(|k|^{-1/2} \hat{\rho}) := \Phi(0\oplus|k|^{-1/2} \hat{\rho})$.
The parameter $\lambda\not=0$ is the coupling constant, which
models the strength of the interaction.
The Hamiltonian for the interacting system is
\begin{equation}\label{Eq:DefHg}
\Hg  =  H_{osc} \tensor \one +  \one \tensor \Haf  + H_I + 
R q^2\tensor \one,
\end{equation}
where $R:= (\lambda^2/2)  \||k|^{-1} \hat{\rho}\|_\h^2$.
The operator $H_I$ and the additional potential 
$R\cdot q^2$, $\hat{\rho}\in \h$ is the coupling function,
form the so-called dipole approximation, which arises from
a model with a minimally coupled Hamiltonian by replacing $e^{\imath kq}$ by $1+ kq$
and by applying a unitary transform, see \cite{Spohn1997}. 
The crucial simplification is, that $\Hg$ is quadratic in annihilation and creation
operators. Note, that we have the representation $\Haf= \int |k|  a^*(k) a(k) d^3k$,
where $a(k)$ and $a^*(k)$ are annihilation
and creation operators for sharp momentum $k$, see
\cite{ReedSimonII1980}. Thus $\Haf$ is quadratic as well.
\begin{Hyp}\label{Hyp1} We assume
\begin{enumerate}
\item $\hat{\rho}(k) > 0,\ k\in \R^3$.
\item $\hat{\rho}$ is rotation invariant.
\item In polar coordinates: $(0, \infty)\ni r\mapsto \hat{\rho}(r) $ 
      has an analytic continuation on $\{ z\in \C :  |\Im  z| \leq  2 \pi  \beta^{-1}\}$,
      also denoted by $\hat{\rho}$.
\item $\sup_{|s| \leq  2 \pi \beta^{-1}} \int |  
       \hat{\rho}(r + \imath s) |^2 (1+|r|^3) dr  < \infty$
\item For the analytic continuation 
      we have $\hat{\rho}(r) = \hat{\rho}(-r),\ r\in \R$.
\end{enumerate}
\end{Hyp}
We write for a functions $f$ with domain in $\R^3$, 
$f(r,\Theta)$ with $r>0$ and $\Theta\in S^2$ if it is
expressed in polar coordinates. Moreover, we write 
$f(r)$, if $f$ is rotation invariant. $S^2$ denotes
in this context the unit sphere.\\ 
\indent The Hamiltonian in \eqref{Eq:DefHg} was analyzed in \cite{Arai1981a,Arai1981c}.
Among other things therein is shown, that 
$\dom(\Hg)= \dom(\Hosc\tensor \one + \one \tensor \Haf)$,
where $\dom$ denotes  the domain of an operator.
An explicit formula for the asymptotic annihilation and creation operators
is given. It can be deduced from \cite{Arai1981a}, that up to a unitary isomorphism
we have $\Hg\cong \Haf+E_{gs}$, where $E_{gs}$ is the ground state
energy of $\Hg$. This unitary isomorphism is a Bogoliubov transform,
that means, it can be defined by introducing new annihilation and creation
operators and a new vacuum vector, see \cite{Berezin1966}. In the following, we will
make use of some definitions and lemmata in \cite{Arai1981a,Arai1981c}, which
are reformulated in Appendix \ref{Sec3} below.\\
\indent We define the dynamics in the Heisenberg picture
on the algebra of observables 
\begin{equation*}
\Ag:= \cl_{\mathcal{B}(\Hig)}
\linhull\big\{ \W{c\oplus f} :  c\oplus f\in \C\oplus \f \big\},
\end{equation*}
with $\f:= \{ f\in \h :  (1+|k|^{-1/2})f\in \h\}$. Note,
that we can extend regular states on $\Ag$ to annihilation- and creation
operators, which are the actual observables of interest.
The time-evolution is given by
\begin{equation*}
A \mapsto \tau_{t}(A)\in \Ag,\quad \tau_{t}(A)
:= e^{\imath t \Hg}A e^{-\imath t \Hg},\quad A\in \Ag.
\end{equation*}
For fixed $t\in\R$ $\taug_t$ is a map from $\Ag$ to $\Ag$, which is
linear and multiplicative, and obeys $\tau_{t}(A^*)=\tau_{t}(A)^*$. Moreover,
$\tau=(\tau_t)_{t\in\R}$ is a group since
$
\tau_t\circ \tau_s=\tau_{t+s},\quad \tau_0=\one.
$
In the sequel this is called a $\ast$-automorphism group.\\
A state is a positive, normed, linear functional on $\Ag$. This definition
covers the vector states in $\Hig$. Furthermore, it allows
to define equilibrium states for positive temperature $\beta^{-1}$.
Since $\Hg$ has continuous spectrum, it is not possible to define
a Gibbs state for $\Hg$. Thus, we follow the lines of Haag, Hugenholz
and Winnink \cite{HaagHugenholzWinnink1967}, who characterized 
equilibrium states by means of the $(\taug,\beta)$-KMS property, 
see Definition \ref{Def:ClustMixEqu} below. An extensive representation
of the theory of operator algebras can be found in the textbooks
\cite{BratteliRobinson1987,BratteliRobinson1996}.\\    	 	
\indent Next, we give definitions in the context of
 operator algebras, see also \cite{Attal2006}. 
\begin{Def}[Quantum Dynamical System]\label{Def:DynSys}
\begin{enumerate}
\item Let $\mathcal{A}$ be a unital $C^*$-algebra, $\tau$ a $\ast$-automorphism group
and $\omega$ a state. The triple $(\mathcal{A}, \tau, \omega)$
is a quantum dynamical system, if
$
\R\ni t\mapsto \omega(A^*\tau_t(A))
$
is continuous for every $A\in \mathcal{A}$.
\item 
Let $\h$ be any Hilbert space and $\mathcal{M}\subset\mathcal{B}(\h)$ be a
unital $C^*$-algebra. $\mathcal{M}$ is a $W^*$-algebra, if it is weakly closed.
$(\mathcal{M},\tau)$ is called a $W^*$-dynamical system, if in addition   
$
\R\ni t \mapsto \tau_t(A),\quad A\in \mathcal{M}
$
is continuous, while $\mathcal{M}$ carries the $\sigma$-weak topology of $\mathcal{B}(\h)$.
\end{enumerate}
\end{Def}
The next definition fixes some dynamical
properties, which are essential for the following.
\begin{Def}[Mixing, Clustering,  Equilibrium State]\label{Def:ClustMixEqu}
Let $(\mathcal{A}, \tau, \omega)$  be a quantum dynamical system.
\begin{enumerate}
\item  $(\mathcal{A}, \tau, \omega)$ is mixing, if
$
\lim_{t\to \infty} \mu(\tau_t(A))= \omega(A)
$
for any $\omega$-normal state $\mu$ and any $A\in \mathcal{A}$.
\item $(\mathcal{A}, \tau, \omega)$ is clustering, if
$
\lim_{t\to\infty}\omega(A \tau_t(B) C) = \omega(A C) \omega(B)
$
for $A, B, C\in \mathcal{A}$.
\item $\omega$ is an equilibrium state for $(\mathcal{A}, \tau)$
      at inverse temperature $\beta$,
      if for any $A, B\in \mathcal{A}$ there is a function
      $F_\beta(\cdot,A,B)$ being analytic in the strip $\{ z\in \C :  0<\Im z< \beta\},$
      continuous on its closure and taking the boundary conditions
      \begin{gather*}\label{Eq:KMSCond}
       F_\beta(A, B, t) = \omg(A \taug_t(B)),\\
       F_\beta(A, B, t + \imath  \beta) = \omg(\taug_t(B) A).
      \end{gather*} 
      In this case $\omega$ is also called a $(\tau, \beta)$-KMS state.
\end{enumerate}
\end{Def}
We will say that a dynamical system has the property of return to equilibrium, 
if it is mixing. Let $(h_\omega,\pi_\omega,\Omega_\omega)$ be the GNS-triple, see
\cite{BratteliRobinson1987}. For $\mu$ being $\omega$-normal means, that
there is a sequence of positive numbers $(p_n)_n$
with $\sum_{n=1}^\infty p_n=1$ and an orthonormal system of vectors
$(\phi_n)_n\subset \h_\omega$ satisfying
$
\mu(A)= \sum_{n=1}^\infty p_n \langle \phi_n
 | \pi_{\omega}[A] \phi_n \rangle_{\h_\omega},\ \forall A\in \Ag.
$
It can be seen using an approximation argument that $(\mathcal{A}, \tau, \omega)$ is
mixing if it is clustering. Let $\disp:= 1+2\varrho$. 
$\varrho(k)= (e^{\beta |k|}-1)^{-1}$ is the momentum density 
in Planck's law.
\begin{Def}[Bose gas system]\label{Def:IntSyst}
Let $\Af:= \We(\f)$. $\tauf$ is defined by
\begin{equation*}
\tauf_t(W):= e^{\imath t \Haf}  W e^{-\imath t \Haf},\quad W\in \Af
\end{equation*}
An equilibrium state is established by
\begin{equation*}
\omf(\W{f}) =\exp( -1/4 \langle f |  \disp f\rangle_\h).
\end{equation*}
We call the quantum dynamical system $(A_f,\tauf,\omf)$ 
Bose gas system.
\end{Def}
Note that $\tauf_t(\W{f})=\W{e^{\imath t|k|}f}$ and $\tauf_t(\Phi(f))=\Phi(e^{\imath t|k|}f)$.
%
%
\section{Return to Equilibrium for the Harmonic Oscillator}\label{Sec4}
\subsection{Statements and Discussions}
In the next theorem we compare the quantum dynamical system
$(\Ag,\tau,\omg)$ with the Bose gas system. We prove that
both are equivalent up to Bogoliubov transform.
\begin{Satz}[ Isomorphism of the Dynamical Systems]\label{Satz:Isom}
There is a unique symplectic map $v : \C \oplus \f\rightarrow \f$,
such that
\begin{equation*}
\gamma :  \Ag\rightarrow \Af,\ \quad \W{c\oplus f}\rightarrow \W{v(c\oplus f)}
\end{equation*}
is a $\ast$-isomorphism and $\omega:= \omf \circ \gamma$ is
an equilibrium state for $(\Ag,\tau)$. Moreover, we have that
\begin{equation}
\omega(\W{c\oplus f})= \exp(-1/4 \langle v(c\oplus f)|\eta\,v(c\oplus f)\rangle)
\end{equation}
 The following diagram
\begin{equation}\label{Eq:Diag2}
\xymatrix{
\Ag     \ar[r]^{\tau_t} \ar[d]^\gamma &\Ag  \ar[d]^\gamma   \ar[r]^{\omg}&\C \\
\Af \ar[r]^{\tauf_t}             &   \Af \ar[ru]_{\omf}\\
}
\end{equation}
is commutative. We call the quantum dynamical system $(\Ag, \taug, \omega)$
the interacting system. Moreover, there is for $t\in \R$ a unique symplectic
map $w_t : \C \oplus \f \rightarrow \C \oplus \f$, such that $\tau_t(\W{c\oplus f})=\W{w_t(c\oplus f)}$
and $w_{t+s}= w_t\circ w_s$.
\end{Satz}
Let $\mathfrak{k}, \mathfrak{k}'$ be  complex pre-Hilbert spaces.
We say that $s : \mathfrak{k}\rightarrow \mathfrak{k}'$ is symplectic,
whenever $s$ is real linear and 
$\Im \langle s(f) | s(f')\rangle_{\mathfrak{k}'}
=\Im\langle  f |  f'\rangle_{\mathfrak{k}},\  f, f'\in \mathfrak{k}.$
Of special interest is the subalgebra $\We( \C\oplus 0)$,
it contains the observables for the particles. In the context
of open quantum systems it is the small system. The excited states
of the small system are represented as states of $\Ag$ defined by
$
\mu(A)=\omega(B^*AB),\quad A\in \Ag
$
for some $B\in \We( \C\oplus 0)$. 
For $f\in L^2(\R^3)$ we denote by
$\tilde{f}$ a function from
$\R\times S^2$ to $\C$, defined by
\begin{equation}\label{Deftilde}
\tilde{f}(r, \Theta)=
\begin{cases}
f(r, \Theta),\ &r \geq 0\\
\ovl{f(-r, \Theta)},&\ r < 0.
\end{cases}
\end{equation}
Let $H^2(\kappa)$ be the real vector space of functions $f : (0, \infty)\rightarrow L^2(S^2)$,
such that $\tilde{f}$ has an analytic continuation on the strip 
$\{ z\in \C :  |\Im z|<\kappa\}$ and obeying
\begin{equation*}
\int_\R  \sup_{\eta\in \R : |\eta|<\kappa}\|\tilde{f}(r+\imath \eta)\|^2_{L^2(S^2)}(1+|r|^3)dr<\infty.
\end{equation*}
$H^2(\kappa)$ is also invariant with respect to complex conjugation.
We also can regard each $f\in H^2(\kappa)$ as a function in $\f$,
in the sense that
$ k\mapsto f(|k|,k\cdot |k|^{-1})$
belongs to $\f$.
Let $ \mathfrak{g}$ be the real vector space of functions $h\in \C\oplus \f$,
such that $v(h)$ belongs to 
\begin{equation}\label{DefR}
\mathcal{R}= \{ \big(a \imath|k|^{1/2}+  b |k|^{-1/2}\big) f \in \f :a,b\in\R, f\in H^2(\kappa)\},
\end{equation}
which is dense in the norm $\|g\|'_+:= \|g\|_{\h}+\||k|^{-1/2}g\|_{ \h}$.
This is proved in Lemma \ref{DensityR} below.
The positive constant $\kappa$ 
will be fixed in Theorem \ref{Satz:RetEq}. $v$ is introduced in Theorem \ref{Satz:Isom}
\begin{Def}[Analytic states and analytic observables]\label{Def:SmallSys} Let
$\Ao:= \linhull\{\W{g}\in \We(\C\oplus \f) : g\in \mathfrak{g}\}\subset \Ag$  be
the subalgebra of analytic observables and let $\mathcal{S}_0$ be 
the analytic states of $\Ag$ defined by
\begin{equation*}
\mathcal{S}_0=\{ \mu  :  \exists B\in \Ao 
\forall A\in \Ag\ \mu(A)=\omg(B^*AB)\}.
\end{equation*}
\end{Def}
The next theorem states the property of return to equilibrium,
and that the decay in time is exponentially fast for states in $\mathcal{S}_0$
and observables in $\Ao$. The proof is given in the subsequent subsection.
Here and in what follows, $\const$ denotes a constant independent of $t$,
its values may change from one line to the other.
\begin{Satz}[ Return to Equilibrium]\label{Satz:RetEq}
For small $\lambda\not=0$ the interacting system
$(\Ag, \tau, \omega)$ is clustering. Moreover,
for $A\in \Ao$ and $\mu\in \mathcal{S}_0$ one has
\begin{equation}\label{ExpDecay}
|\mu(\tau_t(A))-\omg(A)|\leq \const e^{-\kappa  t}
\end{equation}
for some constant and for  $0<\kappa<\hat{\kappa}$. $\hat{\kappa}$ is a fixed decay rate 
$2\pi\beta^{-1}>\hat{\kappa}>0$ depending only on $\lambda$.
$\hat{\kappa}$ is defined in Lemma \ref{Lemma:Zeros} below.
\end{Satz}
\subsection{Auxiliary Statements}
Note, that $D_+$ and  $D$ are defined in Definition \ref{Def1} below. 
In the next lemma the zeros of the analytic continuation of
$\R\ni r \mapsto  D_+(r^2)$ on the strip around the real axis
are determined. The imaginary part of the zeros gives
the  decay rate in Theorem \ref{Satz:RetEq}. The function
$[0,\infty)\mapsto D(r^2)$  already occurs in the analysis
in \cite{Arai1981a,Arai1981c} to determine the life time of resonances for
the temperature zero Hamiltonian $\Hg$. The zeros are
related to Fermi's golden rule for $\Hg$.

\begin{Lemma}\label{Lemma:Zeros}
For $\lambda\not=0$ small enough, there is an analytic continuation
of 
$\R\ni r \mapsto  D_+(r^2)$
on $\{  z\in\C :  |\Im  z| <  |\Im \hat{\kappa}(\lambda)| \}$
that has no zeros. Furthermore, $\hat{\kappa}(\cdot)$ is even, analytic and
\begin{equation*}
\hat{\kappa}(\lambda) =  1 + \kappa_2 \lambda^2 + \kappa_4\lambda^4 + \ldots,
\end{equation*}
where $\Im  \kappa_2 =  2 \pi^2  \hat{\rho}^2(1)$. 
\end{Lemma}

\begin{Lemma}\label{DensityR}
$\mathcal{R}$ is dense in $(\f,\|\cdot\|)$, where $\|g\|'_+:= \|g\|_{\h}+\||k|^{-1/2}g\|_{ \h}$.
\end{Lemma}
This Lemma is used to show, that
the algebra $\Ao$ and
the set of states $\mathcal{S}_0$
are large in the following sense:
\begin{Lemma}\label{Lemma:Density}
$\mathfrak{g}$ is dense in $\C\oplus \f$,
with respect to the norm
$
\|g\|_+:= \|g\|_{\C\oplus \h}+\||k|^{-1/2}g\|_{\C\oplus \h}.
$
Moreover, the image of $\mathfrak{g}$ under $v$ is dense in $(\f,\|\cdot\|'_+)$.
\end{Lemma}
The following Lemma is an ingredient in the proof of
Theorem \ref{Satz:RetEq}, ensuring the exponential
decay. The special choice of form factors in $\mathfrak{g}$
allows to express the real- and the imaginary part
of the scalar products in  \eqref{Eq:ScalarPrDecay} and \eqref{eq2.6.2} below as an integral over $\R$.
The decay is obtained by shifting the contour of integration
in the upper complex half plane. Since in the definition of 
$f, g\in\mathfrak{g}$ the inverse of $D_+(r^2)$ (analytically
continued) occurs, the decay rate is determined by the
zeros of $D_+(r^2)$.
\begin{Lemma}\label{Lemma:Decay}
For $f, g\in \mathfrak{g}$ one has
\begin{gather}\label{Eq:ScalarPrDecay}
|\Re \langle v(f)|\disp  e^{\imath t |k|} v(g) \rangle_\h |
\leq \const e^{-\kappa t}\\ \label{eq2.6.2}
|\Im \langle v(f)|e^{\imath t |k|} v(g)\rangle_\h |\leq \const e^{-\kappa t}.
\end{gather}
$\const $ is a positive constant uniform in $t$.
\end{Lemma}
\subsection{Proofs}
\begin{proof}[Proof of Theorem \ref{Satz:Isom}]
The proof uses intensively results in \cite{Arai1981a,Arai1981c}, quoted in  Appendix \ref{Sec3} below.
From Lemma \ref{lem1.5} follows that
$\Phi_{in}(f) := \frac{1}{\sqrt{2}} \big( a_{in}(f) + a_{in}^*(f) \big)$
is equal to
\begin{eqnarray*}
 \Phi\big( (\langle  \ovl{Q}_+ | f \rangle_{\h}
 + \langle  f | \ovl{Q}_- \rangle_{\h} )\oplus( \ovl{W}_+ f + W_- \ovl{f})\big)
\end{eqnarray*}
$a_{in}(f)$ and $a^*_{in}(f)$ are the incoming annihilation and creation operators.
$a_{in}(f), a^*_{in}(f)$ as well as $W_+,  W_-,  Q_-,  Q_+$ are defined in Appendix \ref{Sec3}. 
Thus we have
$$e^{\imath  t \Hg} \Phi_{in}(f) e^{-\imath  t \Hg} 
= \Phi_{in}(e^{\imath  t |k|} f).$$
Combining Lemma \ref{lem1.3}(\ref{lem1.3d}) and \ref{lem1.3}(\ref{lem1.3e}), such as Equation ~\eqref{Def2e},~\eqref{Def2f} and \eqref{Def1}, we obtain 
\begin{equation}\label{Ident1}
\ovl{W}_+  \ovl{Q}_+=W_-Q_-,\ \ovl{W}_-\ovl{Q}_+=W_+Q_-,\ \|Q\|_\h=1.
\end{equation}
A simple but lengthly calculation using the Identities in \eqref{Ident1} and Lemma \ref{lem1.4} yields 
$\Phi(c\oplus h) =  \Phi_{in}( v(c\oplus h) )$. Here,
$v$ is a real linear operator from $\C\oplus \f$ to $\f$ defined by
\begin{equation}\label{Defv}
v(c\oplus h) := \ovl{W}_+^*h + c \ovl{Q}_+ - \ovl{W}_-^* \ovl{h} - \ovl{c}  \ovl{Q}_-.\\
\end{equation}
By Lemma \ref{lem1.4} the operator $v$ is surjective since 
$
h =  v\big( (\langle  \ovl{Q}_+ | h \rangle_\h 
+ \langle  Q_- | \ovl{h} \rangle_\h )\oplus (\ovl{W}_+h + W_-\ovl{h})\big).
$
Moreover, $v$ is symplectic, since 
\begin{align*}
\Im &\langle  c\oplus h | c'\oplus h' \rangle_{\C\oplus \h} 
 =  -\imath  [ \Phi(c\oplus h), \Phi(c'\oplus h') ]\\ \nonumber
&= -\imath  [ \Phi_{in}( v(c\oplus h) ),\Phi_{in}( v(c'\oplus h') ) ] \\
 &= \Im \langle  v(c\oplus h) | v(c'\oplus h') \rangle _{\h}.
\end{align*}
We deduce that $v$ is injective. From the definition of $\Phi_{in}$ follows that
\begin{equation*}
\taug_t(\Phi(c\oplus h ))
 =  \Phi_{in}( e^{\imath  t |k|} v(c\oplus h) ) =  \Phi( w_t(c\oplus h)  )
\end{equation*}
for the real linear, time dependent operator $w_t$ defined by
\begin{align}\label{Defwt}
w_t&(c\oplus h)=\nonumber
\big( \langle  \ovl{Q}_+| e^{\imath t|k|} v\rangle_{\h}
+\langle e^{\imath t|k|} v|\ovl{Q}_- \rangle_{\h} \big)\\ 
&\oplus 
\big( \ovl{W}_+e^{\imath ,t|k|}v
+W_-e^{-\imath t|k|}\ovl{v}\big),
\end{align}
where $v:=v(c\oplus h)$.
Since $\Phi_{in}( e^{\imath  t |k|} v(c\oplus h) ) = \Phi_{in}( v( w_t(c\oplus h) ) )$
we have that
\begin{equation}\label{Eq4.6}
e^{\imath  t |k|} v(c\oplus h) = v( w_t(c\oplus h) ).
\end{equation}
Since $v$ is symplectic, we may define a $\ast$-isomorphism $\gamma$
by
$
\gamma :  \Ag \rightarrow \Af,\quad \W{c\oplus h}\mapsto \W{v(c\oplus h)}.
$
By Equation \eqref{Eq4.6}, Definition \ref{Def:ClustMixEqu} and Definition \ref{Def:IntSyst} 
we conclude  that $\omg :=  \omf  \circ  \gamma$
is a $(\taug, \beta)$-KMS state over $\Ag$. 
\end{proof}
\begin{proof}[Proof of Theorem \ref{Satz:RetEq}:] First, we prove the exponential decay.
Let $v_i :=  v( f_i)$ for $f_i\in \mathfrak{g},\ i = 1, 2, 3$. Using 
$\taug_t(\W{f_2})=\W{w_t(f_2)}$, the CCR relation for Weyl operators in \eqref{Eq:WeylCCR},
such as \ref{Satz:Isom} and \eqref{Eq4.6}, we obtain
\begin{align*}
\omg\big(&\W{f_1}\taug_t(\W{f_2})\W{f_3}\big)\\ \nonumber
=& \exp\big(-1/4\|\disp^{1/2} 
(v_1+e^{\imath t|k|}v_2+v_3)\|^2_\h\big)\cdot\\ \nonumber
&
\exp\big(-(\imath/2)\Im\{\langle v_1|e^{\imath t|k|}v_2\rangle_\h
+\langle v_1+e^{\imath t|k|}v_2|v_3\rangle_\h\}\big).
\end{align*}
Because of
\begin{align*}
\omg(\W{ f_1}\W{ f_3})
 =& \exp(-(1/4)\|\disp^{1/2}(v_1+v_3)\|^2_\h) \\ \nonumber
&\phantom{\exp}
\cdot \exp(-(\imath /2) \Im \langle v_1|v_3\rangle_h)
\end{align*}
we get
\begin{align}\label{Eq4.24c}
\omg(&\W{f_1})\taug_t(\W{ f_2})\W{ f_3})\\ \nonumber
=& \omg(\W{ f_1}\W{ f_3})\omg(\W{ f_2})\\ \nonumber
%
&\cdot\exp\big(-(1/2)\Re\langle v_1+v_3|\disp
 e^{\imath t|k|}v_2\rangle_\h\big)\\ \nonumber
&
\cdot\exp\big(-(\imath/2)\Im\langle v_1-v_3|e^{\imath t|k|}v_2\rangle_\h\big).
\end{align}
Lemma \ref{Lemma:Decay} implies
\begin{align}\label{Eq:ExpDec}
\big|&\omg\big(\W{f_1})\taug_t(\W{ f_2})\W{ f_3})\\ \nonumber
&\phantom{\big|\omg\big(\W{f_1})}
-\omg(\W{ f_1}\W{ f_3})\omg(\W{ f_2})\big|
\leq\const e^{-\kappa t}.
\end{align}
By definition of $\Ao$ and $\mathcal{S}_0$ the exponentially decay
in \eqref{ExpDecay} follows. To prove the first 
statement, we assume $f_1,  f_2,  f_3\in \f$. As
before we obtain \eqref{Eq4.24c}, but with $v_1,  v_2,  v_3\in \f$.
Since $\wlim_{t\to \infty}e^{\imath t|k|}v_2=0$, we get
\begin{align*}
\lim_{t\rightarrow \infty}\omg(\W{ f_1}\taug_t(\W{ f_2})&\W{ f_3})\\
&= \omg(\W{ f_1}\W{ f_3})\omg(\W{ f_2}).
\end{align*}
Here, $\wlim_{t\to \infty}$ is the weak limit in the Hilbert space sense.
From this we obtain
$
\lim_{t\rightarrow \infty} \omg(A \taug_t(B) C)
 =  \omg(A C) \omg(B)
$,
for $A, B, C\in \linhull\{ \W{f} :  f\in \f\}$. A
density argument yields, that $(\Ag,\tau,\omega)$ is
clustering, and thus mixing.
\end{proof}
\begin{proof}[Proof of Lemma \ref{Lemma:Zeros}]
Let $G$ be defined for $|\Im{z}| < \eta$ and $\lambda \in \C$ by
\begin{align*} \nonumber
G(z,\lambda):=& -z^2+1+\||k|^{-1} \hat{\rho}\|^2_\h \lambda^2+4 \pi^2 \imath \lambda^2 \hat{\rho}^2(z)z\\
&+2\pi
\lambda^2 \int_{-\infty}^\infty \frac{\hat{\rho}^2
(r+\imath \eta)(r+\imath \eta)}{z -(r+\imath \eta)}dr
\end{align*}
Recall that $\hat{\rho}(r)$ is an even function. By definition 
of $D(z)$ in \ref{Def1} we obtain for $\Im z>0$, that
\begin{equation*}
D(z^2)=-z^2+1+\||k|^{-1} \hat{\rho}\|^2_\h\lambda^2 
+2\pi \lambda^2 \int_{-\infty}^\infty \frac{\hat{\rho}^2(r)r}{z-r}dr.
\end{equation*}
The residue theorem yields that $G_{\lambda}(\cdot) := G(\cdot, \lambda)$ 
is an analytic continuation of $D(z^2)$ on the lower half plane. By uniqueness,
and since $D(z^2)$ is even, we get
\begin{equation*}
G(z, \lambda) = G(-z, \lambda) = G(z, -\lambda).
\end{equation*}

Let $s \geq  0$. We can choose
$p_{\epsilon}(s)$, such that 
$ s^2+\imath \epsilon
=p_{\epsilon}(s)^2,\ \Re p_{\epsilon}(s)\geq 0$ 
and $\Im p_{\epsilon}(s)>0$, then
\begin{align*}
G(s,\lambda)&=\lim_{\epsilon \to  0+}G(p_{\epsilon}(s),\lambda)
=\lim_{\epsilon\to 0+}D(s^2+\imath \epsilon)\\
&= D_{+}(s^2).
\end{align*}
Thus $G_\lambda$ is an analytic continuation of $D_+(r^2)$.
For any $0 < \eta' < \eta$ we have 
\begin{equation}\label{Eq4.16}
\sup_{ \{z :  |\Im z|<\eta'\}} | G(z, 0) - G(z, \lambda)|
 \leq  C_{\eta'} |\lambda|^2.
\end{equation}
Since $\partial_z G(\pm  1, 0) =  \mp 2$ 
the implicit function theorem yields two
analytic functions $\kappa_{\pm1}$ in a neighborhood
of zero, with $\kappa_{\pm1}(0)=\pm 1$ and
\begin{equation}\label{Eq4.17}
G(z, \lambda) = 0 \Leftrightarrow  z = \kappa_{\pm1}(\lambda)
\end{equation}
for any  $z $ in a complex neighborhood of $1$ or $-1$, respectively.
By \eqref{Eq4.16} there is a neighborhood of zero, such that \eqref{Eq4.17}
holds for all $z\in \{ z\in \C :  |\Im z| \leq \eta'\}$ for some $\eta'>0$
independent of $\lambda$.
By symmetry of $G$ and uniqueness of $\kappa_{\pm 1}$ we have
$\kappa_{-1}(\lambda) = -\kappa_{+1}(\lambda)$ and $\kappa_{+1}(\lambda) = \kappa_{+1}(-\lambda)$, 
in particular $\partial_{\lambda}^{(2n+1)} \kappa_{+1}(0) = 0$. For the second
derivative we have
\begin{align*}
\partial_{\lambda}^2\kappa_+(0) &=  
-\frac{(\partial_\lambda^2 G)(1, 0)}{(\partial_z G)(1, 0)}
 = \||k|^{-1}  \hat{\rho}\|^2_\h \\
&+ 2 \pi  \mathcal{P} \int_{-\infty}^{\infty}\frac{\hat{\rho}^2(r) r}{1 - r}dr 
 + 4 \pi^2  \hat{\rho}^2(1) \imath,
\end{align*}
where $\mathcal{P} \int_{-\infty}^{\infty}$ 
is the Cauchy-principal value. This completes the proof.
\end{proof}
\begin{proof}[Proof of Lemma \ref{DensityR}]
We will first prove an auxiliary statement: 
\begin{center}
$H^2(\kappa)$ is dense in $\hat{\mathcal{H}}:= L^2( \R^3, (|k|^{-2}+|k|)d^3k).$
\end{center}
\indent For $f\in \hat{\mathcal{H}}$ we have that
\begin{equation}
\int_0^\infty \int_{S^2} |f(r,\Theta)|^2 d\Theta  (1+r^3)dr<\infty.
\end{equation}
Let $\hat{\mathcal{K}}=L^2( (0,\infty), (1+r^3)dr )$. Since
$\hat{\mathcal{H}}=\hat{\mathcal{K}}\tensor L^2(S^2)$ we
may restrict ourselves the case where $f=f_1\cdot f_2$.
Assume first that $f_2$ is real-valued. There is a smooth function $\phi$
with compact support in $(0,\infty)$, and
$\|f_1-\phi\|_{\hat{\mathcal{K}}}\le \epsilon.$
Next, we define a continuation $\psi$ of $\phi$ to $\R$ by $\psi(0):=0$ and $\psi(r):=\ovl{\phi(-r)}$
for $r<0$. Since $\phi$ vanishes near zero, we get $\psi\in C_c^\infty(\R,\C)$.
The Fourier transform $\hat{\psi}$ is a real-valued Schwartz function.
Thus there is a $\eta\in  C_c^\infty(\R,\R)$ such that
$$\big(\|\eta-\hat{\psi}\|^\sim\big)^2:=  \sum_{\stackrel{\nu,\mu\in \N_0}{ \nu,\mu\le 3}} \| (1+|s|^2)^{\nu/2}  \partial^{\mu} (\eta-\hat{\psi})\|_{L^2(\R)}^2\le \epsilon^2.$$
For some universal constant we obtain for the inverse Fourier transform $\check{\eta}$ that $\|\check{\eta}-\psi\|^\sim\le \const \epsilon$.
We obtain for another constant that $\|\check{\eta}-f_1\|_{\hat{\mathcal{K}}}\le \const \epsilon$. It is elementary
to see that $\check{\eta}\cdot f_2$ belongs to $H^2(\kappa)$.
When $-\imath f_2$ is real valued, we modify the proof.
We choose $\phi$ as before. Next choose the continuation to $\R$, such
that $\psi(r)=-\ovl{\psi(-r)}$. Then we can choose $\eta\in  C_c^\infty(\R,\imath \R)$,
thus $\check{\eta(r)}=-\ovl{\check{\eta}(-r)}$. As before, $\check{\eta}\cdot f_2$ belongs to $H^2(\kappa)$.
 This proves the auxiliary statement.\\
\indent Let $g\in \mathcal{C}_c^\infty(\R^3\setminus\{0\},\C)$, then $f:= (a\imath |k|^{1/2}+ b|k|^{-1/2})^{-1}g\in \hat{\mathcal{H}}$. 
Since multiplication with $(a\imath |k|^{1/2}+ b|k|^{-1/2})$ is a continuous map
from $\hat{\mathcal{H}}$ to $(\f,\|\cdot\|_+')$, $\mathcal{C}_c^\infty(\R^3\setminus\{0\},\C)$ is dense in
$(\f,\|\cdot\|_+')$ and $H^2(\kappa)$ is dense in $\hat{\mathcal{H}}$ we conclude the proof.
\end{proof}
\begin{proof}[Proof of Lemma \ref{Lemma:Density}]
The definition of $v$ and $Q$ in \eqref{Defv} and Lemma \ref{Def1} below imply that for $c = a + \imath  b$
$$v(c\oplus 0) =  a |k|^{-1/2} \ovl{Q} + \imath  b |k|^{1/2} \ovl{Q}\in \mathcal{R}.$$
Thus $\C\oplus 0\subset \mathfrak{g}$. Let
$f\in \f$, we need to show that $0\oplus f\in \operatorname{cl}_{\|\cdot\|_+'}\mathfrak{g}$.
From \ref{lem1.4} we deduce that $f = \ovl{W}_+ g + W_- \ovl{g}$
for $g = \ovl{W}_+^* f - \ovl{W}_-^* \ovl{f}\in \f$.\\
\indent Since $\mathcal{R}$, defined in \eqref{DefR}, is dense in $(\f,\|\cdot\|_+)$,
there is a sequence $(g_\nu)_{\nu}\subset \mathcal{R} $ with 
$g_\nu  \to  g$, as $\nu  \to  \infty$. We obtain
\begin{equation*}
f_\nu  :=  \ovl{W}_+ g_\nu + W_- \ovl{g}_\nu  \to  f,\quad \textrm{as}\ \nu  \to  \infty.
\end{equation*}
On account of Lemma \ref{lem1.4} we get $g_{\nu} = v(c_\nu \oplus  f_{\nu})\in v(\mathfrak{g})$ for $c_\nu :=  \langle  \ovl{Q}_+ | g_\nu \rangle_\h  + \langle  Q_- | \ovl{g}_\nu \rangle_\h $. 
We conclude that $c_\nu\oplus  f_{\nu}\in\mathfrak{g}$ and $ 0\oplus f_{\nu} \in \mathfrak{g}$. 
This completes the proof.
\end{proof}
\begin{proof}[Proof of Lemma \ref{Lemma:Decay}]
Let
\begin{gather*}
f :=  (\imath  a |k|^{1/2} + b |k|^{-1/2} ) f'\\
g :=  (\imath  a' |k|^{1/2}  + b' |k|^{-1/2} )  g'
\end{gather*}
and $h$ be defined by
\begin{equation*}
h(r):= \int_{S^2} \ovl{f'(r,\Theta)}g'(r,\Theta)d\Theta,\quad r\geq 0,
\end{equation*}
where $d\Theta$ is the uniform measure on the sphere.
The definition of $f'$ and $g'$  implies that $\tilde{h}$ 
has an analytic continuation on the strip
$\{ z\in \C :  |\Im z|<\kappa \}$, see also \eqref{Deftilde}. 
In order to calculate $\Re  \langle  f | \eta e^{\imath  t  |k|}  g \rangle_\h$
we write the scalar product as an integral over $\R^3$.
Next, we introduce polar-coordinates. After integrating over the
sphere we obtain the following:
\begin{align*}
\Re& \langle f|\disp e^{\imath t |k|} g\rangle_\h\\ \nonumber
&=\Re \int_0^\infty \big(aa'r^3+ bb'r 
+\imath a'br^2 -\imath b'a r^2\big)\\ \nonumber
&
\phantom{=\Re \int_0^\infty \big(aa'r^3+ bb'r}
\cdot \coth(\beta r/2)e^{\imath t r}h(r)dr\\ \nonumber
&=1/2  \int_\R \big(aa'r^3+ bb'r +\imath a'br^2 -\imath b'a r^2\big)\\ \nonumber
&\phantom{=\Re \int_0^\infty \big(aa'r^3+ bb'r}
\cdot\coth(\beta r/2)e^{\imath t r}\tilde{h}(r)dr 
\end{align*}
Since $\tilde{h}$ has an analytic continuation on the strip,
we may shift the contour of integration in the upper half plane. 
This complete the proof of \eqref{Eq:ScalarPrDecay}. The proof of
\eqref{eq2.6.2} follows analogously.
\end{proof}

\section{ Comparison with the Liouvillean Approach}\label{Sec5}
\subsection{Liouvillean approach}
In this section we will sketch the Liouvillean approach.
As mentioned in the introduction, this approach is
widely used to study dynamical properties of small systems, coupled 
to a heat bath, see for instance \cite{JaksicPillet1996a,JaksicPillet1996b,BachFroehlichSigal2000,
Merkli2001,DerezinskiJaksic2001,DerezinskiJaksic2003}.  
For this approach it is not necessary that the
particle is described by a harmonic oscillator,
it is applicable to a broader class of Hamiltonians $\Hael$. 
The key ingredient is the existence of a Gibbs state. 
For this it is sufficient and necessary
that
\begin{equation}
Z_\beta :=  \Tr\{ e^{-\beta  \Hael} \} < \infty,
\end{equation}
for a fixed $\beta>0$. The starting point 
for the model underlying the Liouvillean approach is the $C^*$-algebra
\begin{equation*}
\Aobs=\mathcal{B}(\Hel)\tensor \We(\f).
\end{equation*}
The algebra $\Aobs$ is taken instead of $\Ag$,
since it is left invariant by the dynamics $\tau^0$ of the noninteracting
system for  any choice of $\Hael$. In general, 
the dynamics $\tau^0$ is given by
\begin{equation}\label{Deftau0}
\tau^0_t(A):= e^{\imath t H_0}A e^{-\imath t H_0},\quad A\in \Aobs.
\end{equation}
$H_0$ is the sum of the particle Hamiltonian $\Hael\tensor \one$ and the field Hamiltonian
$\one \tensor \Haf$ multiplied with a tensor-factor. More precisely, 
for the harmonic oscillator in the dipole approximation
$H_0$ is obtained form $\Hg$ by setting $H_I:=0$ in Equation \eqref{Eq:DefHg}
and the particle Hamiltonian is 
\begin{equation*}
\Hael  := 1/2 \big(-\Delta  
+ (1 + \lambda^2  \||k|^{-1} \hat{\rho}\|^2_\h)  q^2 \big),
\end{equation*}
acting in $L^2(\R)$.We deduce from Equation \eqref{Deftau0}, that 
\begin{equation}
\tau^0_t= \tauel_t\tensor \tauf_t,
\end{equation}
where $\tauf_t$ is the free dynamics, see Definition
\ref{Def:IntSyst}. $\tauel$ is dynamics for the particle system, given by 
$\tauel_t(B)= e^{\imath t \Hael}B e^{-\imath t \Hael},\ B\in \mathcal{B}(\Hel)$.
On $\Aobs$ define a state $\omega_0$ by
\begin{equation*}
\omega_0(B\tensor \W{f})
 :=  \big(Z_\beta^{-1}  \Tr\{B e^{-\beta  \Hael}\}\big)
\cdot \omf(\W{f}).
\end{equation*}
We recall, that the first factor is the Gibbs state for $\Hael$ at inverse 
temperature $\beta$, and the second is the equilibrium state for 
the Bose gas system.  The $(\tau^0,\beta)$-KMS property of $\upsilon_0$ in the
sense of Definition \ref{Def:ClustMixEqu} can be verified by
a short calculation. \\
\indent In a second step, one makes an explicit 
GNS-construction $(\Kg, \pi_0, \Omo)$. Let
\begin{equation*}
\Kg :=  L^2(\R)\tensor L^2(\R)\tensor 
\mathcal{F}_b[\mathfrak{h}\oplus \mathfrak{h}]
\cong \mathcal{F}_b[\C\oplus \C \oplus \mathfrak{h} \oplus \mathfrak{h}]
\end{equation*}
and for $X\in \mathcal{B}(\Kg)$
\begin{equation*}
\Omo  =  Z^{-1/2}_\beta   k_{\beta/2}\tensor \Omega_{\h\oplus\h},\qquad \upsilon_0(X):= \langle \Omo\,|\, X \Omo\rangle
\end{equation*}
where $k_{\beta/2}$ is the integral kernel of 
$e^{-\beta/2  \Hael}$ in $ L^2(\R^2)\cong L^2(\R)\tensor L^2(\R)$.
The $*$-isomorphism $\pi_0$ is given by
\begin{equation*}
\pi_0[A\tensor \W{f}]
 =  A\tensor \one \tensor \W{(1+\varrho)^{1/2} f\oplus \varrho^{1/2} \ovl{f}}.
\end{equation*}
Let $\Mg :=  \pi_0 [ \Aobs ]''$ be the 
bicommutant of $\pi_0 [ \Aobs ]$, it is the weak 
closure of $\pi_0 [ \Aobs ]$ in $\mathcal{B}(\Kg)$. 
To define automorphism groups on $\Mg$ we introduce 
the operators
\begin{align*}
\Lo&:=  \Lel\tensor \one\tensor \one  +  \one\tensor \one\tensor \Lf\\
\Lel&:= \Hael\tensor \one  - \one\tensor\Hael
\end{align*}
$\Lf$ is defined as $\Haf$ in Equation \eqref{DefHf}, with
the difference that $\h$ is replaced by $\h\oplus\h$,
and $h_{ph}$ is given by $|k|\oplus (-|k|)$.
Since $\Omo$ is cyclic and separating for $\Mg$ we can define
the modular conjugation by
\begin{equation}
\Jg  X \Omo  = X^*\Omo,\quad X\in \Mg.
\end{equation}
As a consequence of the $(\tau^0,\beta)$-KMS property we obtain 
for the commutant $\Mg' = \Jg \Mg \Jg$.\\
\indent In general, there is no reason that  $\taug$ should leave $\Aobs$
invariant. However, we define a dynamics $\alpha$ on $\Mg$ by 
\begin{equation*} 
\tev{t}{X} :=  e^{\imath  t \Lg} X e^{-\imath  t \Lg},\quad X\in \Mg.
\end{equation*}
Here, $\Lg$ is the so-called \emph{Standard Liouvillean} defined by
\begin{equation}\label{DefLg}
\Lg :=  \Lo  + Q_I -  \Jg  Q_I \Jg.
\end{equation}
$Q_I$ describes the interaction of the particle with the heat bath.
In the case of the dipole approximation $Q_I$ is given by
\begin{equation}\label{DefQI}
Q_I :=  \lambda  q\tensor \one \tensor  
\Phi( (1+\varrho)^{1/2}|k|^{-1/2} \hat{\rho}\oplus \varrho^{1/2}|k|^{-1/2} \hat{\rho}).
\end{equation}
The reason for the choice of $Q_I$ is the subsequent formal calculation
\begin{align}\label{DysonSeries}
\pig[\taug_t(B\tensor \Phi(f))]
&= \sum_{j=0}^\infty
\pig\big([ \ldots, [\imath \Hg , B\tensor \Phi(f)] \ldots]\big)\\ \nonumber
&=\sum_{j=0}^\infty[ \ldots, [\imath \Lg , \pig\big(B\tensor \Phi(f)\big)] \ldots]\\ \nonumber
&=\alpha_t\big(\pig[B\tensor \Phi(f)]\big)
\end{align}
Note, that there is no contribution in Equation \eqref{DysonSeries} from $\Jg  Q_I \Jg$, and that
$\pig\big(B\tensor \Phi(f)\big)=B\tensor \one \tensor \Phi( (1+\varrho)^{1/2}f\oplus \varrho^{1/2}\ovl{f}).$
Surely, one can define $Q_I$ that is obtained from a different interaction, as long as the calculation \eqref{DysonSeries} remains true.\\
\indent Once we have $\Lg$, one can show self-adjointness, the invariance of
$\Mg$ with respect to $\alpha_t$.
Furthermore, it can be proved 
that
\begin{equation*}
\Omg  = c   e^{-(\beta/2) (\Lo + Q_I)} \Omo,
\end{equation*}
is cyclic for $\pi_0[\Aobs]$, separating for $\Mg$, and 
normed for some $c>0$. To show that $\Omo\in \dom(e^{-(\beta/2) (\Lo + Q_I)})$
one needs  that $ |\lambda |$ is sufficiently small.
For more details see\cite{Koenenberg2009b}.\\
\indent Furthermore $\Omg$ is in the kernel of $\Lg$ and
\begin{equation}\label{Defupsilon}
\upsilon(X) :=   \langle  \Omg | X \Omg \rangle_{\Kg}
\end{equation}
and $\upsilon$ defines an $(\alpha,\beta)$-KMS-state on $\Mg$.
\indent It was shown in \cite{JaksicPillet1996a, JaksicPillet1996b}, 
that the dynamical properties of  $(\Mg,\alpha,\Omg)$ as the 
mixing property or the ergodicity are encoded in the spectrum of $\Lg$. 
But an analysis of the spectrum of $\Lg$
has not been done for the harmonic oscillator, yet. 
For a short introduction in this topic we refer the reader to \cite{Attal2006}.\\
\indent However, in the case of the harmonic oscillator in the dipole approximation,
the algebra is $\Aobs$ is left invariant by $\tau$, and the question arises
for a rigorous proof of the identity in \eqref{DysonSeries}.
Let $j$ be the canonical embedding of $\Ag$ into $\Aobs$.
We will show
\begin{Lemma}\label{LemmaExtension}
The diagram
\begin{equation}\label{Diag3}
\xymatrix{
\Ag \ar[d]^{\taug_t} \ar[r]^{j} & \Aobs  \ar[r]^{\pi_0} & \Mg \ar[d]^{\tevabb{t}} \\ 
\Ag  \ar[r]^{j} & \Aobs \ar[r]^{\pi_0} & \Mg 
}
\end{equation}
is commutative.
\end{Lemma}
This diagram illustrates, that $\alpha$ extends the dynamics
$\tau$ on $\Mg$. The proof is given in Subsection \ref{ProofsSection3}.\\
\indent Since we have for small $0<|\lambda|\ll 1$ two $(\tau,\beta)$-KMS states on $\Ag$, namely
$\omega:= \omf \circ \gamma$ and $\upsilon\circ \pi_0\circ j$, the question
arises if both coincide. We will prove
\begin{Satz}\label{Thm5.1}
For small $0<|\lambda|\ll 1$ we have that
\begin{equation*}
\omg = \upsilon\circ \pi_0 \circ j,\qquad \textrm{on } \Ag. 
\end{equation*} 
\end{Satz}
This is illustrated in the next diagram
\begin{equation*}
\xymatrix{
\Ag \ar[r]^{\pi_0\circ j } \ar[rd]_{\omg}&\Mg \ar[d]_{\upsilon} \\
&\C
}
\end{equation*}
In Subsection  \ref{ProofsSection3} we will give a proof of Theorem \ref{Thm5.1}. 
Next, we introduce an $W^*$-dynamical system for the Bose gas.
This is essential to formulate our results.\\
\indent Let $\Kf =  \Fock_b[\h\oplus\h]$ and 
$\Omf :=  \Omega_{\h\oplus\h}$. 
The Araki-Woods-Representation  $\pi_{AW} : \Af\rightarrow \mathcal{B}(\Kf)$
is given by
\begin{equation*}
\pi_{AW}(\W{g}) =  \W{(1 + \varrho)^{1/2} g\oplus \varrho^{1/2} \ovl{g}}.
\end{equation*}
It can be shown that $(\mathcal{K}_f,\pi_{AW},\Omf)$ is a GNS-triple
for $\Af$ and $\omf$.\\
Let $\Mf :=  \pi_{AW}[\Af]''$ and
$\upsilon_f(X) =  \langle  \Omf | X \Omf\rangle_{\Kf}$,
and 
$\tevf{t}{X} = e^{\imath  t  \Lf} X e^{-\imath  t  \Lf},\ X\in \Mf$.
$\Omf$ is cyclic for $\pi_{AW}[\Af]$ and separating for $\Mf$. 
The diagram is commutative
\begin{equation*}
\xymatrix{
\Af \ar[r]^{\tauf_t} \ar[d]_{\pi_{AW}}& \Af    \ar[r]^{\omf}  \ar[d]_{\pi_{AW}}&\C \\          
\Mf \ar[r]^{\alpha^f_t}& \Mf \ar[ru]_{\upsilon_f}
}
\end{equation*}
in this setting.
We say that $(\Mf,\alpha^f)$ is the $W^*$-dynamical system for the
Bose gas. For a proof of this statements see \cite{Attal2006}.\\
\indent Moreover, we have an embedding of $(\Ag,\tau,\omg)$ into
$(\Mg,\alpha,\upsilon)$ by means of $\pi_0\circ j$ and an embedding
into $(\Mf,\alpha^f,\upsilon_f)$ by means of $\pi_{AW}\circ \gamma$.
In fact, both embeddings are unitarily equivalent. We first show 
\begin{Satz}\label{Thm5.2}
There is an isometric isomorphism $U : \Kf \rightarrow  \Kg$,
such that $U e^{\imath  t \Lf}  =   e^{\imath  t \Lg} U$ and $U \Omf =  \Omg$. Let
$\gamma_U :  \mathcal{B}(\Kf)\rightarrow \mathcal{B}(\Kg),\quad
\gamma_U(A)= U A U^{-1}$. Then $\gamma_U \circ \pi_{AW}\circ \gamma =  \pi_0 \circ j $.
\end{Satz}
The statement of Theorem \ref{Thm5.2} is illustrated by
\begin{equation*}
\xymatrix{
(\Kf,\Omf)  \ar[r]^{ e^{\imath t \Lf}}\ar[d]^{ U}  & (\Kf,\Omf)\ar[d]^{ U}\\
(\Kg,\Omg) \ar[r]^{ e^{\imath t \Lg}}& (\Kg,\Omg)
}\qquad
\xymatrix{
\Ag \ar[r]^{\pi_0\circ j} \ar[d]^\gamma & \mathcal{B}(\Kg) \\
\Af \ar[r]^{\pi_{AW}} &\mathcal{B}(\Kf) \ar[u]^{\gamma_U}
}
\end{equation*}
A key ingredient in the proof of Theorem \ref{Thm5.2} is
\begin{Lemma}\label{lem5.1}
$(\pi_0\circ j)[\Ag]'' = \Mg$
\end{Lemma}
Thus all elements of $\Mg$ can be approximated by elements of $(\pi_0\circ j)[\Ag]$.
From this we obtain:%
\begin{Kor}\label{Cor5.3}
\begin{enumerate}
\item The $W^*$-dynamical systems $(\Mg, \alpha, \Omg)$
      and $(\Mf, \alpha^f, \Omf)$ are unitarily equivalent.
\item $\Lg$ is unitarily equivalent to $\Lf$ and
      $\dom(\Lg) =  U \dom(\Lf)$.
\item $\sigma(\Lg) =  \R,\ \sigma_{sc}(\Lg) = \emptyset,\ \sigma_{ac}(\Lg) = \R,\ 
      \sigma_{pp}(\Lg) = \{0\}$ and $\Omg$ is up to scalar multiples
      the only vector in the kernel of $\Lg$.           
\end{enumerate}
\end{Kor}
\begin{proof}
\begin{enumerate}
\item Since $ \gamma(\Ag) =  \Af$ and
$ (\pi_0 \circ j) [\Ag]\subset \Mg$, we have
\begin{equation*}
\gamma_U : \pi_{AW}[\Af] \rightarrow  \Mg.
\end{equation*}
Furthermore, since $\Mf$ (resp. $\Mg$  ) is the $\sigma$-weak closure of
$\pi_{AW}[\Af]$ (resp. $(\pi_0 \circ j) [\Ag]$), and $\gamma_U, \gamma_U^{-1}$
are  $\sigma$-weakly continuous, we conclude that
$\gamma_U : \Mf\rightarrow \Mg$ is a spatial $\ast$-isomorphism.
\item follows from $e^{\imath  t \Lg} = U e^{\imath  t \Lf} U^{-1}$.
\item The statements are well known for $\Lf$ and $\Omf$.
\end{enumerate}
\end{proof}
\subsection{Proofs.}\label{ProofsSection3}
The proof of Lemma \ref{LemmaExtension} is based on the fact that both
$\Hg$ and $\Lg$ are quadratic in the field operators.
It follows that $\taug_t(\W{c\oplus f})=\W{w_t(c\oplus f)}$ and that
$\alpha_t(\W{ c\oplus c'\oplus h\oplus h})=\W{ \tilde{w}_t(c\oplus c'\oplus h\oplus h)}$
for some real linear operator $w_t$ acting on $\C\oplus \f$ and
some real linear operator $\tilde{w}_t$ acting on $\C\oplus\C\oplus \h\oplus \h$. 
The proof of this follows from a theorem of Berezin for annihilation and
creation operators, see \cite{Berezin1966}.
In the proof of Lemma \ref{LemmaExtension} we compare $\tilde{w}_t$ with $w_t$
on a subspace of $\C\oplus \C \oplus \h \oplus \h$.
\begin{proof}[Proof of Lemma \ref{LemmaExtension}]
It is sufficient, to check that 
\begin{equation}\label{EqualEq1}
(\pi_0\circ j)[\tau_t(\W{c\oplus f})]=\alpha_t((\pi_0\circ j)[\W{c\oplus f}]).
\end{equation}
By Theorem \ref{Satz:Isom} we have that
\begin{align}
(\pi_0\circ j)&[\tau_t(\W{c\oplus f})]
=\mathcal{W}(w_t(c\oplus f)^{(1)}\oplus 0 \\ \nonumber
&\oplus (1+\varrho)^{1/2}w_t(c\oplus f)^{(2)}\oplus \varrho^{1/2}\overline{w_t}(c\oplus f)^{(2)})
\end{align}
with $w_t(c\oplus f)=:w_t(c\oplus f)^{(1)}\oplus w_t(c\oplus f)^{(2)}$. Since $\pig\circ j$ is
a regular representation, we only need to check Equality \eqref{EqualEq1} for
field operators instead of Weyl operators.
Since $\Lg$ is quadratic in creation- and annihilation operators,
there is a vector $\tilde{w}_t(c\oplus c' \oplus g\oplus g')$
such that
\begin{align}\label{Equality1}
\Phi(\tilde{w}_t(c\oplus c' \oplus g\oplus g')\big)=\alpha_t \big(\Phi(c\oplus c' \oplus g\oplus g')\big)
\end{align}
Moreover, $t\mapsto \tilde{w}_t$ is a strongly continuous one-parameter
group of real linear operators in $\C \oplus\C \oplus \h \oplus \h$. The
domain of its generator $\tilde{A}$ is
$\C \oplus\C\oplus\dom(|k|)\oplus\dom(|k|).$
Moreover,
\begin{equation}\label{Commutator}
\Big[\imath \Lg , \Phi(c\oplus c' \oplus g\oplus g')\Big]
=\Phi(\tilde{A}(c\oplus c' \oplus g\oplus g')).
\end{equation}
On the other side for field operators in $\mathcal{F}_b[\C\oplus \h]$,
$t\mapsto w_t$ defines a strongly continuous one-parameter group of real linear 
operators. Let $A$ be its generator. We have
\begin{equation}
[\imath \Hg , \Phi(c\oplus f)]=\Phi(A(c\oplus f)),
\end{equation}
for $c\oplus f\in \C\oplus \dom(|k|)$. Thus it suffices to 
show that
\begin{align}\label{Identity}
\tilde{A}&(c\oplus 0 \oplus (1+\varrho)^{1/2}f\oplus \varrho^{1/2}\ovl{f})
=A(c\oplus f)^{(1)}\oplus 0 \\ \nonumber
&\qquad \qquad \oplus (1+\varrho)^{1/2}A(c\oplus f)^{(2)}\oplus \varrho^{1/2}\ovl{A(c\oplus f)}^{(2)}
\end{align}
where $A(c\oplus f)=A(c\oplus f)^{(1)}\oplus A(c\oplus f)^{(2)}$.
Since $\Jg \Qg \Jg$ makes no contribution to Equation \eqref{Commutator} for
$c'=0,\ g=(1+\varrho)^{1/2}f$ and $g'=\varrho^{1/2}f$ we can verify \eqref{Identity}
using \eqref{Eq:DefHg} and \eqref{DefLg}.
\end{proof}
The proof of Theorem \ref{Thm5.1} is the most technical 
in this section, while the idea is simple. Since
$\Omega$ is in the kernel of $\Lg$ and by Lemma \ref{LemmaExtension} 
we deduce, that
\begin{equation}
\nu\circ \pig \circ j=\nu\circ\alpha_t\circ \pig \circ j
=\nu\circ \pig \circ j\circ\taug_t.
\end{equation}
$\nu$  is normal with respect to $\nu_0$, by \eqref{Defupsilon}.
The proof of Theorem \ref{Thm5.1} will be completed, if we show
that we have $\lim_{t\to\infty}\mu\circ \alpha_t\circ \pig \circ j=\omg$ 
for all $\nu_0$-normal states $\mu$.
\indent In the next statement we deduce an explicit formula for
$\upsilon_0\circ (\pi_0\circ j)$.
\begin{Lemma}\label{LemNonIntKMS}
For $c\oplus f\in \C\oplus \f$ we have that
\begin{equation*}
\upsilon_0\big( (\pi_0\circ j)(\W{c\oplus f})\big)
=\exp\big(-1/4\|\dispo^{1/2}(c\oplus f)\|^2_\h\big),
\end{equation*}
where
$
\varrho_0(c\oplus f)=\big(\exp( \beta \alpha)-1 \big)^{-1}
\big(\Re(c)\alpha^{-1/2}+\imath\Im(c)\alpha^{1/2}\big)
\oplus \varrho f
$
and
$
\alpha^2=1+\lambda^2 \||k|^{-1}\hat{\rho}\|_\h
$
\end{Lemma}
\begin{proof}
By definition of $\nu_0$ it suffices to
show that
\begin{gather*}
 (\exp( \beta \alpha)-1 ) Z_\beta^{-1}  \Tr\{\W{c} e^{-\beta  \Hael}\}
=\\
\big(\Re(c)\alpha^{-1/2}
+\imath\Im(c)\alpha^{1/2}\big).
\end{gather*}
Let us first introduce ladder operators for the harmonic
oscillator
\begin{equation*}
B^*:= \frac{\alpha^{1/2} q - \imath \alpha^{-1/2} p}{\sqrt{2}},
\quad
B:= \frac{\alpha^{1/2} q + \imath \alpha^{-1/2} p}{\sqrt{2}}
\end{equation*}
defined on the Schwartz functions. There is up to a normalization
constant unique vector  $\Omega_{\alpha}:= (\alpha  \pi^{-1})^{1/4} e^{-\alpha q^2/2}$,
such that  $B\Omega_{\alpha}=0$. We also have
$[ B, B^* ] = \one_{S(\R)}$
and
$
\Hael:=\alpha B^*B+\alpha/2.
$
A complete system of eigenvectors for $\Hael$ is given by
\begin{equation*}
(n!)^{-1/2} (B^*)^n\Omega_{\alpha},\ n=0,1,2,\ldots
\end{equation*}
corresponding to the eigenvalues $E_n=\alpha  n+\alpha/2$.
Furthermore, we obtain that
$
\W{c}= \exp(\imath (c' B^*+ \ovl{c'}B)/\sqrt{2}),
$
where 
$c':= \Re(c)\alpha^{-1/2}+\imath\Im(c)\alpha^{1/2}$. 
The lemma follows by a general theorem about Gibbs state
of operators defined by  second quantization, see for instance
\cite[Prop. 5.28]{BratteliRobinson1996}.
\end{proof}
The subsequent Lemma is essential to prove the
convergence of $\upsilon_0$-normal states to $\omega$.
\begin{Lemma}\label{Conv}
We have for $A,B\in \Ag$ that
\begin{equation}\label{eq2}
\lim_{t\to\infty}\upsilon_0\big( (\pi_0\circ j)(B^*\tau_t(A)B)\big) 
= \upsilon_0\big( (\pi_0\circ j)(B^*B)\big) \cdot \omega(A)
\end{equation}
\end{Lemma}
Before we will give a proof of Lemma \ref{Conv} we show
\begin{Lemma} \label{Lem:Comm}
The following operators, defined as quadratic forms, on $\f$
can be extended to compact operators on $\h$
\begin{equation*} 
[ \disp^{1/2},  W_{\pm}^*],\quad [ \disp^{1/2},  W_{\pm}].
\end{equation*}
Moreover, there exists compact operators $\mathfrak{k}_1,$ und $\mathfrak{k}_2$ with
\begin{align}\label{eq: CompEst1}
W^*_- \disp W_- &=\disp^{1/2}\mathfrak{k}_1\disp^{1/2}\\ 
W^*_+ \disp W_+ &= \disp\label{eq: CompEst2}
+ \disp^{1/2}\mathfrak{k}_2\disp^{1/2}, 
\end{align}
regarded as quadratic forms on $\f$.
\end{Lemma}
\begin{proof}
First, assume that $g, h$ are some sufficiently regular multiplication operators on
$L^2(\R^3)$. Recall the Definition of $G_\epsilon$ and $G$ in  \ref{Def2} and 
 \ref{lem1.3}.
On $\f$ we define the operator as a quadratic form
\begin{equation*}
K_\epsilon:= g [ \disp^{1/2},  G_\epsilon]  h .
\end{equation*}
This operator has  the integral kernel
\begin{equation*}
K_\epsilon(k,k') 
= \frac{g(k)}{|k|^{1/2}}
\frac{\disp(k)^{1/2}-\disp(k')^{1/2}}{k^2-(k')^2-\imath \epsilon}
\frac{h(k')}{|k'|^{1/2}}.
\end{equation*}
By definition of $\varrho$ in Definition \ref{Def:IntSyst} we obtain
\begin{align*}
\Big|&\frac{\disp(k)^{1/2}-\disp(k')^{1/2}}{|k|-|k'|}\Big|\\ \nonumber
&\quad\leq \frac{|\varrho(k)-\varrho(k')|}{\big||k|-|k'|\big|}
\disp(k)^{-1/4}\disp(k')^{-1/4}\\ \nonumber
&\quad\leq \const \disp(k)^{3/4}\disp(k')^{3/4}.
\end{align*}
Thus for $\alpha\in \{0,1\}$ we get
\begin{equation*}
\int |K_\epsilon(k,k') |^2d^3k d^3k'
\leq \const \Big\| \frac{\disp^{3/4}g}{|k|^{1/2+\alpha}}\Big\|^2_\h 
\cdot \Big\| \frac{\disp^{3/4}h}{|k|^{3/2-\alpha}}\Big\|_\h^2.
\end{equation*}
Thus $K_\epsilon$ extends to a Hilbert-Schmidt operator for sufficiently
regular $g,h$.
Applying dominated convergence to the integral kernels of $K_\epsilon$,
we obtain $K:=\lim_{\epsilon\to 0}K_\epsilon$
in the Hilbert-Schmidt-norm $\|\cdot\|_{HS}$ and that 
$K$ is an extension of $g [ \disp^{1/2},  G]  h$.
By the definitions of $T^*$ and $W_{\pm}^*$ given in \ref{Def2}
we obtain
\begin{equation*}
\|g [ \disp^{1/2},  T^*]  h\|_{HS}
\leq \const \Big\| \frac{\disp^{3/4}\hat{\rho}  g}{|k|^{\alpha}}\Big\|_\h
\cdot
\Big\| \frac{\disp^{3/4}\hat{\rho}  h}{|k|^{1-\alpha}}\Big\|_\h,
\end{equation*}
and for $g=h=1$
\begin{equation}\label{eq2a}
\|[ \disp^{1/2},  W_{\pm}^*]\|_{HS}
\leq \const \Big\| \frac{\disp^{3/4}\hat{\rho} }{|k|^{1/2}}\Big\|_\h
\cdot
\Big\| \frac{\disp^{3/4}\hat{\rho}}{|k|^{1/2}}\Big\|_\h.
\end{equation}
By Hypothesis \ref{Hyp1} we deduce that right side of \eqref{eq2a} is finite.
The some upper bound holds for $W_{\pm}^*$ is replaced by $W_{\pm}$.
Since $W^*_-W_-$ is compact, see Lemma \ref{lem1.4}, we obtain Equation \eqref{eq: CompEst1}.
Equation \eqref{eq: CompEst2} follows  
analog to Lemma \ref{lem1.4}.
\end{proof}
\begin{proof}[Proof of Lemma \ref{Conv}]
By Lemma \ref{LemNonIntKMS} we obtain
\begin{equation*}
\upsilon_0\big( (\pi_0\circ j)(\W{f})\big)
=\exp\big(-1/4\|\dispo^{1/2}f\|^2_\h\big)
\end{equation*}
for $f\in \C\oplus \f$. Using $\tau_t(\W{f})=\W{w_t(f)}$
and  the CCR  for Weyl operators we obtain 
\begin{align*}
\upsilon_0&\big( (\pi_0\circ j)(\W{f}\tau_t(\W{g})\W{h})\big) \\ \nonumber
&= \upsilon_0\big( (\pi_0\circ j)[\W{f}\W{h}]\big) 
\cdot \exp(-1/4\| \dispo^{1/2}w_t(g)\|^2_\h)\\ \nonumber
&\phantom{= \upsilon_0\big(}\cdot \exp\big( \imath /2 \Im\langle w_t(g)|f-h\rangle_\h\\ \nonumber
&\phantom{= \cdot \exp\big( \imath /2 \Im\langle }
               -1/2 \Re\langle f+h|\dispo  w_t(g)\rangle_\h\big)
\end{align*}
By the Riemann-Lebesgue Lemma we get that  $\wlim_{t\to\infty} e^{\imath t|k|}v(g)
=\wlim_{t\to\infty} \disp^{1/2}e^{\imath t|k|}v(g)=0$ for $g\in \C\oplus\f$. 
From \eqref{Defwt} and Lemma \ref{Lem:Comm} follows that $\lim_{t \to \infty} \langle f\pm h|(1+2\varrho_0)w_t(g)\rangle=0$.
Moreover, for any compact operator $\mathfrak{k}$  and any $u\in \h$ we have
that $\lim_{t\to\infty}\|\mathfrak{k}e^{-\imath  t |k|} u\|_\h=0$.
Thus  Lemma \ref{Lem:Comm} implies
\begin{align*}
\lim_{t \to \infty}&\upsilon_0\big( (\pi_0\circ j)(\W{f}\tau_t(\W{g})\W{h})\big) \\ \nonumber
=&\upsilon_0\big( (\pi_0\circ j)(\W{f}\W{h})\big) \cdot \\ \nonumber
&\lim_{t \to \infty} \exp\big(-\| \disp^{1/2}
( \ovl{W}_+e^{\imath t|k|}v(g)\\ \nonumber
&\phantom{\lim_{t \to \infty} \exp\big(-\| \disp^{1/2}
( \ovl{W}_+}
+W_-e^{-\imath t|k|}\ovl{v(g)})\|^2_\h/4\big)\\ \nonumber
= &\upsilon_0\big( (\pi_0\circ j)(\W{f}\W{h})\big) \exp\big(-1/4\| \disp^{1/2}v(g)\|_\h^2\big)
\end{align*}
As the linear combinations of Weyl operators are dense in $\We(\C\oplus \f)$ 
we infer the convergence in Equation \eqref{eq2}.
\end{proof}
\begin{proof}[Proof of Theorem \ref{Thm5.1}]
Let $\phi :=  (\pi_0\circ j)[B] \Omo\in \Kg$ 
with $\|\phi\|^2_{\Kg} = \omo(B^* B) = 1$.
\begin{align*}
|\langle \Omg|
& (\pi_0\circ j)[A]\Omg\rangle _{\Kg}
-\langle \phi|e^{\imath t\Lg}(\pi_0\circ j)[A]e^{-\imath t\Lg}
\phi \rangle_{\Kg}|
\\ \nonumber
&= |\langle \Omg| e^{\imath t\Lg}
(\pi_0\circ j)[A]e^{-\imath t\Lg}\Omg\rangle _{\Kg}\\ \nonumber
&\phantom{ |\langle \Omg| e^{\imath t\Lg}(\pi_0\circ j)[A]}
-\langle \phi|e^{\imath t\Lg}(\pi_0\circ j)[A]e^{-\imath t\Lg}
\phi \rangle_{\Kg}|\\ \nonumber
&\leq  2\|\phi-\Omg\|_{\Kg}\cdot \|A\|_{\Ag}.
\end{align*}
Next, because of Lemma \ref{Conv}
\begin{align*}
\lim_{t\to\infty}&
\langle \phi|e^{\imath t\Lg}(\pi_0\circ j)[A]
e^{-\imath t\Lg}\phi \rangle_{\Kg}
= \lim_{t\to\infty} \omo(B^*\taug_t(A)B)\\ 
&= \omg(A).
\end{align*}
Hence
\begin{equation*}
|\langle  \Omg |
  (\pi_0\circ j)(A) \Omg \rangle _{\Kg}
-\omg(A)| \leq  2  \|\phi - \Omg\|_{\Kg} 
\cdot  \|A\|_{\Ag}.
\end{equation*}
On the other hand, by Lemma \ref{lem5.1}, we have
\begin{equation*}
\cl  (\pi_0\circ j)[\mathcal{A}] \Omg 
= \cl  (\pi_0\circ j)[\mathcal{A}]'' \Omg 
= \cl  \Mg  \Omg =  \Kg.
\end{equation*}
Therefore, since $\|\phi - \Omg\|_{\Kg}$ can be chosen arbitrarily small,
$\upsilon((\pi_0\circ j)(A)) = \omg(A)$ follows.
\end{proof}
\begin{proof}[Proof of Lemma \ref{lem5.1}]
Obviously, we have $(\pi_0\circ j)[\Ag]''\subset\Mg$.
We need to show that $\Mg$ is contained in the weak closure of $(\pi_0\circ j)[\Ag]$.\\
\indent Let $X\in \Mg$, $\phi\in \Kg$ and $\epsilon>0$.
By the bicommutant-theorem there is a $X'\in \pi_0[\Aobs]$
such that $\|X'\phi-X\phi\|<\epsilon$. Using linearity and
density arguments we may assume that
\begin{equation*}
X'= A\tensor \one \tensor \big(\W{\sqrt{1+\varrho}f\oplus \sqrt{\varrho}\ovl{f}}\big),\quad A\in \mathcal{B}(\Hel)
\end{equation*}
and $\phi=\phi_1\tensor \phi_2 \tensor \phi_3\tensor \phi_4$.
Since $\We(\C)'' =  \mathcal{B}(\Hel)$, the bicommutant-theorem yields
the existence of $W\in \We(\C) $ such that $\|A \phi_1 - W \phi_1\|_{\Kg}<\epsilon$.
Thus
$
\|Y\phi-X\phi\|<2\epsilon
$
for $Y := (\pi_0\circ j)[ W\tensor \W{f} ]$. 
\end{proof}
\begin{proof}[Proof of Theorem \ref{Thm5.2}] Let $\mathcal{H}_1 :=  (\pi_{AW}\circ \gamma)[\Ag] \Omf$ 
and $\mathcal{H}_2 :=  (\pi_0\circ j)[\Ag] \Omg$. Since
$\Omf$ is separating for $(\pi_{AW}\circ \gamma)[\Ag]$, we can define 
$U :  \mathcal{H}_1\rightarrow \mathcal{H}_2$ by
\begin{equation*}
(\pi_{AW}\circ \gamma)[A] \Omf\mapsto (\pi_0\circ j)[A] \Omg,\qquad A\in \Ag.
\end{equation*}
We observe that
\begin{align*}
\langle(&\pi_0\circ j)[A] \Omg | (\pi_0\circ j)[B] \Omg\rangle_{\Kg}
 = \upsilon( \pi_0\circ j)[A^*B]\\\nonumber
& = \omg(A^*B)
= \omf((\pi_{AW}\circ \gamma)[A^*B])\\
& = \langle(\pi_{AW}\circ \gamma)[A] \Omf | (\pi_{AW}\circ \gamma)[B] \Omf \rangle_{\Kf}.
\end{align*}
Therefore, $U$ is an isometric isomorphism. Moreover,
\begin{align*}
e&^{\imath  t \Lg}
 U(\pi_{AW}\circ \gamma)[A] \Omf
 = e^{\imath  t \Lg} (\pi_0\circ j)[A] \Omg\\ \nonumber
&= e^{\imath  t\Lg} (\pi_0\circ j)[A]
 e^{-\imath  t \Lg} \Omg
 = (\pi_0\circ j)[\taug_t(A)] \Omg\\\nonumber
&= U (\pi_{AW}\circ \gamma) [\taug_t(A)] \Omf \\
&= U e^{\imath  t \Lf}(\pi_{AW}\circ \gamma)[A] e^{-\imath  t \Lf} \Omf\\
& = U e^{\imath  t \Lf} (\pi_{AW}\circ \gamma)[A] \Omf.
\end{align*}
That is $e^{\imath  t \Lg} U = U e^{\imath  t \Lf}$ on $\mathcal{H}_1$ . We can
extend $U$ to an isometric map from $\Kf$ onto $\Kg$,
since $\cl  \mathcal{H}_1 =  \cl   (\pi_{AW}\circ \gamma)[\Ag] \Omf = 
 \cl  \pi_{AW}[\We(\h)] \Omf = \Kf$ and 
 $\cl \mathcal{H}_2 = \Kg$. This completes the proof.
\end{proof}
\section{Anharmonic Oscillator}\label{Sec6}
The goal of this section is to prove the property of "return to
equilibrium" for a slightly modified model. We will study the
behavior of dynamics of an anharmonic oscillator
coupled to a heat bath for large times. The starting point is the
automorphism group
\begin{equation*}
t\mapsto \tau^V_t(A) := e^{\imath t H_V}A e^{-\imath t H_V},\ A\in \Ag,
\end{equation*}
where $H_V := \Hg+V$. We call $H_V$ the anharmonic Hamiltonian.
$V$ is a bounded real-valued potential over $\R$. The class of admissible potentials
are Fourier transforms of measures, that is
$
V(q)= \int_\R \nu(d\mu) e^{\imath q  \mu}.
$
More precisely, $\nu$ is a complex-valued Borel measure, such
that $\nu(A)=\ovl{\nu(-A)}$ for any Borel set $A\subset \R$ 
and $-A=\{ -a :  a\in A\}$. Let $|\nu|$ be the absolute value
of $\nu$, we assume that
$
a_i := \int_\R |\nu|(d\mu) |\mu|^i<\infty,\quad i=0, 1, 2
$
and that
\begin{equation}
\label{AssumV}
\kappa > 2(a_0+\tilde{\kappa} a_2).
\end{equation}
The number $\tilde{\kappa}$ depends on $\lambda$ and is
given by
\begin{equation*}
\tilde{\kappa}:=2 \pi \int_\R \frac{\lambda^2  \big|\hat{\rho}(r + \imath  \kappa )^2(r + \imath  \kappa)^2\big|}
{\big|G_\lambda(r + \imath  \kappa) \ovl{G_\lambda(r - \imath  \kappa)}\big|} dr.
\end{equation*}
Note that Condition \eqref{AssumV} only makes sense for fixed $\lambda\not=0$.
This inhibits to prove Fermi's golden rule for the time decay.\\
\indent By means of Equation \eqref{Eq:DefFeldOp} we obtain
$
V= \int_\R \nu(d\mu) \W{\mu\oplus 0}
$
in the sense of operators.\\
Note that $\tau^V_t(A)$ need not to be an element of $\Ag$ for $A\in \Ag$.
So we require an extension to a larger dynamical system. The first step
in this direction is to consider the Hamiltonian $H_P$ and 
$\tau^P_t(A):= e^{\imath t H_P}A e^{-\imath t H_P},\quad A\in \Af$
with $H_P := \Haf+P$ and
$
P := \int_\R  \nu(d\mu)  \W{v(\mu\oplus 0)}\in \mathcal{B}(\Hf).
$
This definition is motivated by the following formal calculations:
\begin{enumerate}
\item $P=\gamma(V)$
\item $\gamma(\tau^V_t(A))= \gamma(e^{\imath t H_V}e^{-\imath t \Hg})\cdot
\tauf_t(\gamma(A))\cdot\gamma(e^{\imath t \Hg}e^{-\imath t H_V})$.
\item Using  Dyson's expansion we get 
\begin{align*}
\gamma&(e^{\imath t H_V}e^{-\imath t \Hg})\\
&= \one + \sum_{n=1}^\infty\int_{0\leq t_n\leq  \ldots \leq t_1\leq t} \unl{dt}  
\gamma\big(\tau_{t_n}(V)\ldots \tau_{t_1}(V)\big)\\
&= \one + \sum_{n=1}^\infty\int_{0\leq t_n\leq  \ldots \leq t_1\leq t}  \unl{dt} 
\tauf_{t_n}(P)\ldots \tauf_{t_1}(P)\\
&=e^{\imath t H_P}e^{-\imath t \Haf}
\end{align*}
Likewise we obtain 
$$\gamma(e^{\imath t \Hg}e^{-\imath t H_V})
=e^{\imath t \Haf}e^{-\imath t H_P}.$$
\item All together we get formally $\gamma(\tau^V_t(A))= \tau^P_t\big(\gamma(A)\big)$
\end{enumerate}
In a second step we define a dynamics $\alpha^Q$ on the $W^*$-algebra
$\Mf$. The automorphism group is given by
\begin{equation}\label{DefAlphaQ}
\taup{t}{X}:= e^{\imath t \LQ}Xe^{-\imath t \LQ},\quad X\in \Mf
\end{equation} 
$\LQ$ denotes the Standard Liouvillean for the
dynamics $\alpha^Q$, it is given by
$
\LQ =  \Lf + Q - \Jg Q \Jg,
$
with
$
Q := \int_\R  \nu(d\mu)  \piaw[\W{v(\mu\oplus 0)}]\in \Mf.
$
$\Jg$ is the modular conjugation corresponding to $\Omf$.
Note that $Q\in \Mf$ and since $\Mf$ is a $W^*$-algebra, we
obtain $\alpha_t^Q(X)\in \Mf$ for $X\in\Mf$.
The definition of $\alpha^Q$ can be motivated in the same way as that 
of $\alpha$ in Section \ref{Sec5}.
Moreover, a $(\alpha^Q, \beta)$-KMS state is given by
$\omp$, where 
$\omp(A) :=  
\langle \Omp | A  \Omp \rangle_{\Kf}
,\quad \Omp:=c e^{-\beta/2  (\Lf + Q)} \Omf.$
$c>0$ is a normalization constant. We remark that
\begin{align}\nonumber
\taup{t}{A} &= \tevf{t}{A}
 + \sum_{n=1}^\infty  \imath^n  \int_{0 \leq  t_n \leq \ldots t_1 \leq  t} d\unl{t} 
[ \tevf{t - t_1}{Q}, \\ \label{CommutatorExp} 
&\qquad
\ldots [ \ldots [ \tevf{t - t_n}{Q}, \tevf{t}{A} ] \ldots ] ],
\end{align}
where $d\unl{t}=dt_1 \cdots dt_n$.
For the formulation of our results we set:
\begin{Def} We define
\begin{enumerate}
\item The observables of the small systems are $(\piaw \circ \gamma)[\We(\C\oplus 0)]$.
\item The analytic observables are $\Mfa:= (\piaw \circ \gamma)[\Ao]$.
\item The analytic states are given by 
      \begin{equation*}
      \mathcal{S}_a:=\{ \mu :  \exists A\in  \Mfa\ \forall X\in \Mf\
      \mu(X) = \langle A \Omf |  X  A \Omf\rangle \} 
      \end{equation*}
\end{enumerate}
\end{Def}
We will prove the following:

\begin{Satz} \label{Lem5.4}
The $W^*$-dynamical system $(\Mf,\alpha^Q,\omp)$ is mixing.
Moreover, we have for $\mu \in \mathcal{S}_a$ and $ A\in \Mfa$
\begin{equation*}
\big| \mu(\taup{t}{A})-\omp(A)\big|
 \leq  \const  \exp(-(\kappa - 2 (a_0 +\tilde{\kappa}  a_2)) t),
\end{equation*}
where $\const$ is some constant depending on $\mu, A$  and $Q$, and for
$\tilde{\kappa}$  defined in \eqref{AssumV}.
\end{Satz}
\subsection{Auxiliary Statements}
The proof of Theorem \ref{Lem5.4} is splitted in three parts.
The first is
\begin{Lemma} \label{Lem5.4b}
For $A,B,C\in \Mfa$ we have
\begin{align}\label{Eq:AnhDecayEst}
|\upsilon_f(A\taup{t}{B}C&)-\upsilon_f(AC) 
\tilde{\upsilon}(B)|\\ \nonumber
&\leq \const \exp(-(\kappa-2(a_0+\tilde{\kappa} a_2))t),
\end{align}
where $\tilde{\upsilon}(B):=\lim_{t\to \infty}\upsilon_f(\taup{t}{B})$.
\end{Lemma}
The second part is
%
%
\begin{Lemma}
$(\Mfa)''=\Mf$.
\end{Lemma}
%
%
\begin{proof}
First, we observe
$
\Mfa=\linhull\{ \W{\sqrt{1+\varrho}f\oplus \sqrt{\varrho}\ovl{f}} :  f\in v(\g)\}.
$
Since $v(\g)$ is dense in $\f$ by Lemma \ref{Lemma:Density}, there exists
for every $f\in \f$ a sequence $(f_n)_n \subset v(\g)$,
so that 
$$\sqrt{1+\varrho}f_n\oplus \sqrt{\varrho}\ovl{f}_n\to
\sqrt{1+\varrho}f\oplus \sqrt{\varrho}\ovl{f},\quad n\to \infty,$$
where the limit is taken with respect to the norm of $\h\oplus \h$.
It follows, that
$$\slim_{n\rightarrow\infty}\W{\sqrt{1+\varrho}f_n\oplus \sqrt{\varrho}\ovl{f}_n}=
\W{\sqrt{1+\varrho}f\oplus \sqrt{\varrho}\ovl{f}}.$$
We conclude that
$
(\Mfa)''= 
\linhull \{ \W{\sqrt{1+\varrho}f\oplus \sqrt{\varrho}\ovl{f}} :  f\in \f\}''= \Mf.
$
\end{proof}

The third part is

\begin{Kor}\label{Kor1g}
For every $\omp$-normal state $\mu$ over $\Mf$ 
and every $C\in \Mf$ we obtain
$\lim_{t\to\infty}\mu(\taup{t}{C})= \omp(C).$
Hence $(\Mf,\alpha^Q,\omp)$ is mixing.
\end{Kor}
\subsection{Proofs.}
In the following $\const$ denotes a constant independent of
$t>0$. The value of $\const$ may change from line to line.
In this subsection the norms and scalar products 
without subscript belong to  $\Kg$ or to $\mathcal{B}(\Kg)$.
\begin{proof}[Proof of Lemma \ref{Lem5.4b}.]
First, we prove Estimate \eqref{Eq:AnhDecayEst} 
for 
$A:=\piaw [\W{f}]$, $B:=\piaw[\W{g}]$,$C:=\piaw[\W{h}]$
for $ f, g, h\in v(\g).$
Note, that $e^{\imath \mu q}=\W{\mu\oplus 0}$ and
let $v_k:=v(\mu_k\oplus 0)$. Recalling the statements of Diagram \eqref{Eq:Diag2} and \eqref{Diag3} 
we get
\begin{equation*}
\tevf{t - t_k}{P}=  \int_\R  \nu(d\mu_k)  \piaw[\W{ e^{\imath (t-t_k)|k|v_k}}].
\end{equation*}
Employing the CCR for Weyl operators we obtain
\begin{equation*}
[\W{f}, \W{g}] = 
(-2\imath)  \sin\big((1/2)\Im \langle  f | g \rangle_\h \big) \W{f+g}.
\end{equation*}
The commutator expression in Equation \eqref{CommutatorExp}
equals
\begin{align*} 
c_n\cdot& \piaw\big(\W{e^{\imath t|k|} u_n}\big)
= 
[ \piaw\big(\W{e^{\imath  (t-t_1) |k|} v_1}\big), [\ldots\\ \nonumber
&\ldots[ \piaw\big(\W{e^{\imath (t-t_n) |k|} v_n}\big), 
\piaw\big(\W{e^{\imath  t |k|} g}\big) ]\ldots ] ]
\end{align*}
with
\begin{align*}
c_n(\unl{t}, \unl{\mu}) &:=  
(-2 \imath)^n  \prod_{k=1}^n \sin\Big(\Im  \langle  e^{-\imath  t_k |k|} v_k 
| u_{k,n} \rangle_\h/2  \Big)
\\ \nonumber
u_{k,n}(\unl{t},\unl{\mu}) &:=  \sum_{m=k+1}^n e^{ -\imath  t_m |k|} v_m + g,
\quad u_{n,n}(\unl{t},\unl{\mu}) := g\\ \nonumber
u_n(\unl{t},\unl{\mu}) &:=u_{0,n}(\unl{t}, n, \unl{\mu}).
\end{align*}
Now and in the sequel we omit the arguments of $c_n, u_{k,n},  u_n$.
Let $\nu(d\unl{\mu}):= \nu(d\mu_1)\tensor \ldots \tensor  \nu(d\mu_n)$
the $n$-fold product measure of $\nu$. In the following we will use the
same symbol for different $n$. The equation in \eqref{Eq:AnhDecayEst} reads
\begin{align}\label{Eq:Sum1}
\upsilon_f&(A \taup{t}{B} C)
= \omf(\W{f} \tau_t^f(\W{g}) \W{h}) \\ \nonumber
&\phantom{=}
+ \sum_{n = 1}^\infty \imath^n 
\int_{0 \leq  t_n  \leq \ldots t_1 \leq  t} d\unl{t} \\ \nonumber
&\phantom{=}
\cdot \int \nu(d\unl{\mu}) 
c_n\cdot \omf\Big(\W{f} \W{e^{\imath t|k|} u_n} \W{h}\Big).
\end{align}
Moreover, for the expectation in $\omf$ in the line before we find
\begin{align}\label{Eq:Delta1}
\omf&\Big(\W{f} \W{e^{\imath t|k|} u_n} \W{h}\Big)\\ \nonumber
&=
 e^{\Delta_n} \omf\big(\W{f} \W{h}\big) \omf\big(\W{u_n}\big),
\end{align}
with
\begin{align*}\nonumber
\Delta_n(t,\unl{t},\unl{\mu}) :=& 
 -(1/2) \Re \langle  f + h
 | \disp e^{\imath  t |k|} u_n \rangle_\h\\ 
&-(\imath /2) \Im \langle f - h | 
e^{\imath  t |k|} u_n\rangle_\h.
\end{align*}
%
Note, that
$ |e^{\Delta_n} - 1| 
\leq  |\Delta_n|  e^{|\Re \Delta_n|}$ 
and 
\begin{equation}\label{Eq:Delta2}
\Big|\omf\big(\W{f} \W{h}\big)
 \omf(\W{ u_n})
 e^{|\Re \Delta_n|}\Big| \leq  1.
\end{equation}
Using  Equation \eqref{Eq:ExpDec}  we obtain an exponentially fast 
decay with rate $\kappa$ for 
\begin{equation*}
\omf(\W{f} \tau_t^f(\W{g}) \W{h})-\omf(\W{f} \W{h})\omf(\tau_t^f(\W{g})).
\end{equation*}
Moreover, an explicit expression for $\upsilon_f(\taup{t}{B})$ is obtained 
by setting $f=0=h$ in \eqref{Eq:Sum1}. Combining this with \eqref{Eq:Delta1}
and \eqref{Eq:Delta2} we obtain for some constant
\begin{align}\label{Eq:Sum2}
|\upsilon_f&(A \taup{t}{B} C) - \upsilon_f(A C) 
\upsilon_f(\taup{t}{B})|\\ \nonumber
&\leq  \const e^{-\kappa  t} 
 + \sum_{n = 1}^\infty  \int_{0 \leq  t_n \leq \ldots t_1 \leq  t} d\unl{t} 
\int |\nu|(d\unl{\mu}) 
\big|c_n \Delta_n\big|.
\end{align}
Here, $|\nu|$ is the absolute value of the measure $\nu$.
The constant depends only on $f, g, h$.
Recalling Lemma \ref{Lemma:Decay} and 
the definition of $\Delta_n$ and $u_n$
\begin{equation} \label{Eq:EstDelta}
|\Delta_n|
 \leq \const \big(\sum_{i=1}^n e^{-(t - t_i) \kappa}  |\mu_i| 
+ e^{-t \kappa}\big).
\end{equation}
for some constant depending on $f, g, h$. 
Furthermore, estimate \ref{LemA1} yields
\begin{equation}\label{Eq:Estc}
|c_n|  \leq  \const e^{-\kappa  t_i} |\mu_i| 
\prod_{k = i+1}^{n} (1 + \tilde{\kappa}  \mu_k^2),\ i=1,2,\ldots,n
\end{equation}
for some constant. Note, that $\tilde{\kappa}$ is independent of $f, g, h$.
We insert Estimate \eqref{Eq:EstDelta} in \eqref{Eq:Sum2}, depending
on the summands in \eqref{Eq:EstDelta} we use a different $i$ in
Estimate \eqref{Eq:Estc}. Thus
\begin{align*}
|\upsilon_f&(A \taup{t}{B} C)-\upsilon_f(A C) \upsilon_f(\taup{t}{B})|\\ \nonumber
&\leq  \const  e^{-\kappa  t} \Big(1 + 
e^{-\kappa  t} \sum_{n = 1}^\infty  
\int_{0 \leq  t_n \leq \ldots t_1 \leq  t} d\unl{t}\\ \nonumber
& \phantom{ \const  e^{-\kappa  t} \Big(1 +}
\int |\nu|(d\unl{\mu})  \big( \sum_{i = 1}^n \mu_i^2  
\prod_{k = i+1}^{n}(1 + \tilde{\kappa}  \mu_k^2)
 + 1\big)\Big)\\ \nonumber
&\leq \const e^{-\kappa  t}\Big( 1  + 
\sum_{n = 1}^\infty  
\frac{t^n}{n!} 
\int |\nu| (d\unl{\mu}) \prod_{k = 1}^{n} (1 + \tilde{\kappa}  \mu_k^2)\Big)
\\ \nonumber%
&\leq 
\const \exp(-t (\kappa - a_0 - \tilde{\kappa}  a_2))
\end{align*}
To finish the proof we need to show that
the convergence $\lim_{t\rightarrow \infty}\upsilon_f(\taup{t}{B})$ 
is exponentially fast. Therefore we shall
compare $\upsilon_f(\taup{t}{B})$ and $\upsilon_f(\taup{s}{B})$ for $0<t<s$.
Again, recall the $\upsilon_f(\taup{t}{B})$ is calculated in
\eqref{Eq:Sum1} ( for $f=0=h$). Using that $\omf $ is $\tauf$-invariant 
we obtain
\begin{align*}
|\upsilon_f(\taup{s}{B}) &- \upsilon_f(\taup{t}{B})|
\leq   \sum_{n=1}^\infty 
\int_{0 \leq  t_n \leq \ldots \leq  t_1 \leq  s} d\unl{t}\\ 
&\one[t_1 \geq  t] 
\int |\nu|(d\unl{\mu}) |c_n| |\omf(\W{u_n})|.
\end{align*}
Since $\W{u_n}$ is a unitary we get $|\omf(\W{u_n})|\leq 1$.
Employing Estimate \eqref{Eq:Estc} for $i=1$ we get
\begin{align*}
\big|&\upsilon_f(\taup{s}{B}) - \upsilon_f(\taup{t}{B})\big|
\\ \nonumber
&\leq \sum_{n=1}^\infty \const 2^n 
\int_{0\leq t_n\leq \ldots \leq t_1\leq s }d\unl{t}\ \one[t_1\geq t] 
(a_0+\tilde{\kappa} a_2)^{n-1}e^{-\kappa t_1}\\ \nonumber
&=  \sum_{n=1}^\infty \const2^n \int_t^s
\frac{(r(a_0+\tilde{\kappa} a_2))^{n-1}}{(n-1)!}e^{-\kappa r} dr\\ \nonumber
&\leq \const \int_t^\infty\exp(-(\kappa-2(a_0+\tilde{\kappa} a_2))r)dr\\ \nonumber
&\leq \const \exp(-(\kappa-2(a_0+\tilde{\kappa} a_2))t).
\end{align*}
We define $\tilde{\upsilon}(B) := \lim_{t\to \infty} \upsilon_f(\taup{t}{B})$ 
using the Cauchy-criterion.
\end{proof}
\begin{proof}[Proof of Corollary \ref{Kor1g}.]
Let now $\phi \in \Kf,\ \|\phi\|=1$ and $A,  B\in \Mfa$, so that $\|A\Omg\|=1$.
\begin{align*}
|\langle \phi  &|  \taup{t}{B}  \phi\rangle
 -  \tilde{\upsilon}(B) |\\ \nonumber
&\leq  |\langle  \phi |  \taup{t}{B}  \phi \rangle
 - \langle  A \Omf | \taup{t}{B} A \Omf \rangle|\\ \nonumber
&\quad+ |\omf(A^* \taup{t}{B} A) - 
 \tilde{\upsilon}(B) |\\ \nonumber
&\leq 2\|A \Omf-\phi\| \cdot \|B\| + |\omf(A^* \taup{t}{B} A)
 - \tilde{\upsilon}(B)|.
\end{align*}
We obtain directly
\begin{align*}
\limsup&_{t\to\infty}|\langle  \phi |  \taup{t}{B}  \phi \rangle
 -  \tilde{\upsilon}(B)| \\ \nonumber
&\leq 2  \|B\|  \inf\{\|A \Omf - \phi\| |\ 
A\in (\piaw\circ \gamma)[\Ao],\\ \nonumber
&\phantom{\leq 2  \|B\|  \inf\{\|A} \|A \Omf\| = 1\} = 0.
\end{align*}
Choosing $\phi =  \Omp$ we obtain by time-invariance of KMS states,
that $\omp  =  \tilde{\upsilon}$ over $\Mfa$.
Assume now $C\in \Mf,\  (C_n)_n \in\Mfa$ and $D\in \Mf'$,
so that $C = \slim_{n \to\infty}C_n$. Since $\LQ \Omp=0$ we have that
\begin{align*}
\langle  D \Omp &|  \taup{t}{C - C_n}  D\Omp \rangle
 = \langle  D^* D \Omp  |  \taup{t}{C - C_n}  \Omp \rangle\\ \nonumber
 =& \langle  D^* D \Omp |  e^{\imath  t  \LQ}(C - C_n)  \Omp\rangle.
\end{align*}
Hence,
\begin{align*}
&\limsup_{t\to \infty}| \omp(C) - \langle D \Omp |  \taup{t}{C} D \Omp \rangle |
\\ \nonumber
&\leq \limsup_{t\to \infty} 
|\omp(C_n) - \langle  D \Omp 
|  \taup{t}{C_n}  D \Omp \rangle|\\ \nonumber
&\quad+ \limsup_{t\to \infty} 
|\omp(C - C_n) - \langle  D \Omp | 
 \taup{t}{C - C_n}  D \Omp \rangle|\\ \nonumber
&\leq (1 + \| D^* D\|)\cdot \|(C - C_n) \Omp\| .
\end{align*}
For $n\to \infty$ we obtain
$\omp(C) = \lim_{t\to\infty}\langle  D \Omp | 
 \taup{t}{C}  D \Omp \rangle.$
Since $\Omp$ is separating for $\Mf$, it
is cyclic for $(\Mf)'$. An approximation argument yields
$
\omp(C) = \lim_{t\to\infty} \langle  \phi |  \taup{t}{C}  \phi\rangle,
$
for all $\phi \in \Kf$. Note, that any normal state $\mu$ over $\Mf$ has a
vector representative $\phi \in \Kf$, since $\omp$ is a KMS state.
\end{proof}
\bigskip
\noindent \textbf{Acknowledgments}\\
This paper is part of the author's PhD requirements.
I am grateful to Volker Bach and Edgardo Stockmeyer for many useful discussions
and helpful advice. The main part of this work 
was done during the author's stay at the Institut for Mathematics at
the University of Mainz. The work has been partially supported by the
DFG (SFB/TR 12).

\appendix
\section{A summary of results and definitions in \cite{Arai1981a,Arai1981c}} \label{Sec3}
In this section we recall definitions and statements, that
are made in \cite{Arai1981a,Arai1981c}. The most important
result is the formula for the asymptotic creation and
annihilation operators given in Lemma \ref{lem1.5}. It is
the starting point of our analysis.
\begin{Def}\label{Def1}
For $z\in \C\setminus [0, \infty)$ we get
\begin{align*} 
D(z)& := -z + 1 + \lambda^2 \||k|^{-1} \hat{\rho}\|^2_\h
 + \lambda^2  \int d^3k  \frac{\hat{\rho}(k)^2}{z - k^2}\\ \nonumber
D_\pm (r)& :=  \lim_{\epsilon  \rightarrow  0+}  D(r + \imath  \epsilon),\quad r\in [0, \infty)\\ \nonumber
Q(k)& := -\lambda  \frac{\hat{\rho}(k)}{D_+(k^2)} \\ 
Q_\pm(k) & :=  (1/2) (|k|^{1/2} \pm |k|^{-1/2}) Q(k). 
\end{align*}
\end{Def}
\noindent Furthermore, we define the following operators on $\h$.
\begin{Def}\label{Def2}
\begin{align} \label{Def2a}
(G_{\epsilon} g)(k) &:=  \int \frac{g(k')}{(|k|  |k'|)^{1/2}(k^2 - k'^2 + \imath  \epsilon)} d^3k'. \\
\label{Def2b}
G &:=  \lim_{\epsilon \rightarrow  0+} G_{\epsilon}\\ 
\label{Def2c}
T g&:=  g + \lambda  |k|^{1/2} Q G |k|^{1/2}  \hat{\rho}  g\\
\label{Def2d}
T^* g&:=  g - \lambda  |k|^{1/2} \hat{\rho} G |k|^{1/2} \ovl{Q} g \\
\label{Def2e}
W_+ g&:=  (1/2) \big\{|k|^{-1/2} T^* |k|^{1/2} + |k|^{1/2} T^* |k|^{-1/2}  \big\} g \\ 
\label{Def2f}
W_- g&:=  (1/2) \big\{|k|^{-1/2} T^* |k|^{1/2} - |k|^{1/2} T^* |k|^{-1/2}  \big\} g
\end{align}
\end{Def}
For the functions $D$ and $D_\pm$ we summarize some
results in
\begin{Lemma}  \label{lem1.1}
\begin{enumerate}
\item $D$ is analytic in $\C  \setminus  [0, \infty)$,
\item $D_{\pm}(s) :=  \lim_{\epsilon \to  0+}  D(s + \imath  \epsilon)$ 
        exists and is continuous for $s\in [0, \infty)$,
\item $\inf_{s\in [0, \infty)}  |D_{\pm}(s)| > 0$,
\item $|D(z) + z| < c_1$ and $|D(z)| > c_2$ for all $z\in\C \setminus [0, \infty)$ 
      and real constants $c_1$ and $c_2$.
\end{enumerate}
\end{Lemma}
Let $M_\alpha(\R^3) = \{ f :  \|f\|_\alpha =  \| |k|^{\alpha} f\|_{\h} < \infty\}$, 
for $\alpha \in \R$. For the operators introduced in Definition \ref{Def2}
we have:
\begin{Lemma} \label{lem1.2}
\begin{enumerate}
\item $G_\epsilon$ is bounded on $\h$, uniformly for $\epsilon>0$.
\item $G  :=  \slim_{\epsilon \to  0+}  G_\epsilon$ exists as an operator on $\h$.
\item $G $ is bounded on $\h$ and $M_{-1/2} (\R^3)$.
\item $G^* = -G$, i.e $G$ is skew-symmetric on $\h$.
\end{enumerate}
\end{Lemma}
Given a bounded operator $A$ on $\h$ we denote by $\ovl{A}$ an operator acting on $g\in \h$ by
means of $(\ovl{A} g)(k) :=  \ovl{(A \ovl{g})(k)}$. The bar is of course the complex
conjugation.
\begin{Lemma} \label{lem1.3}
\begin{enumerate}
\item \label{lem1.3a} $T$ and $T^*$ (see Definition \ref{Def2}) are bounded on 
       $M_\alpha(\R^3)$ for $\alpha = 1/2, 0, -1$.
\item \label{lem1.3b} $T^*$ is the adjoint of $T$.
\item \label{lem1.3c} For a rotation invariant function $h$ on $\R^3$, we have $T^* h T = \ovl{T}^* h \ovl{T}$
\item \label{lem1.3d} Furthermore, if  $h Q\in \h $, then $T^* h Q = \ovl{T}^* h \ovl{Q}$.
\item \label{lem1.3e} $T^* Q = 0$ and $\|Q\|_\h=1$
\end{enumerate}
\end{Lemma}
The next algebraic relations ensure that
the incoming creation- and annihilation operators  fulfill
the CCR.
\begin{Lemma} \label{lem1.4}
The operators $W_+$ and $W_-$ defined in Definition \ref{Def2}
are bounded on $M_\alpha(\R^3)$ for $\alpha = -1/2, 0$ and 
fulfill
\begin{gather*}
W_+^* W_+ - W_-^* W_- + P_{+}- P_{-} = \one,\\
W_+ W_+^* - \ovl{W}_- \ovl{W}_-^* = \one,\\
\ovl{W}_+^* W_- - \ovl{W}_-^* W_+ + P_{+-} - P_{-+} = 0,\\
W_- W_+^*-\ovl{W}_+ \ovl{W}_-^* = 0,
\end{gather*}
where
\begin{gather*}
P_{\pm} f =  \langle  Q_\pm | f \rangle_\h   Q_\pm,\qquad
P_{+-} f =  \langle  Q_- | f \rangle_\h   \ovl{Q}_+ \\
P_{- +} f =  \langle Q_+ | f\rangle_\h   \ovl{Q}_- .
\end{gather*}
Furthermore, $W_-$ is a Hilbert-Schmidt operator with integral kernel
\begin{equation*}
W_-(k,k') =  \frac{\lambda  \hat{\rho}(k) \ovl{Q(k')}}{2 (|k| |k'|)^{1/2} (|k| + |k'|)}.
\end{equation*}
\end{Lemma}
The starting point of our work is the following result:
\begin{Lemma} \label{lem1.5}
Let $\Ho$ be $\Hg$ with $\lambda=0$.
The asymptotic creation- and annihilation- operators $a^{\#}_{in}(f)$ 
exist for $f\in M_0(\R^3)\cap M_{-1/2}(\R^3)$,
\begin{equation*}
a^{\#}_{in}(f)  =  \slim_{t \to -\infty} e^{\imath  t  \Hg } 
  e^{-\imath  t  \Ho }  a^{\#}(f) e^{\imath  t  \Ho } e^{-\imath  t  \Hg },
\end{equation*}
$\dom(a^{\#}_{in}(f))\supset \dom(\Hg)$
for
\begin{align*} 
a_{in}(f) 
 &=  \langle  Q_- | \ovl{f} \rangle_\h  A^* + \langle  Q_+ |  \ovl{f} \rangle_\h  A
 + a^*(W_- \ovl{f}) + a(\ovl{W}_+ f)\\ \nonumber
a^*_{in}(f) 
 &=  \langle  \ovl{f} | Q_- \rangle_\h  A + \langle  \ovl{f} |  Q_+ \rangle_\h  A^*
 + a(W_- \ovl{f}) + a^*(\ovl{W}_+ f).
\end{align*}
Moreover, $a_{in}(f)$ and $a^*_{in}(g)$ fulfill the CCR for $f,g\in \h$.
\end{Lemma}
\section{}\label{Sec7}
The next lemma was originally proved in \cite{Maassen1983}.
\begin{Lemma} \label{LemA1}
Let  $f,f_1,\ldots, f_n$ be vectors in a Hilbert space $\h$,
and real numbers $0 \leq  t_1 \leq  t_2 \leq\ldots t_n$
and $\lambda_1, \lambda_2, \ldots, \lambda_n\in \R$
and $\alpha, \beta,  \gamma > 0$, such that
\begin{align*}
|\Im  \langle f_k |  f\rangle | &\leq   \alpha  |\lambda_k| e^{-t_k  \gamma},\quad k=1,\ldots,n\\ 
|\Im  \langle f_k |  f_j\rangle | &\leq  \beta  |\lambda_k \lambda_j|\cdot e^{-(t_k - t_j)  \gamma}, \quad n\ge k > j \ge  1.
\end{align*}
Then for $j=1,\ldots,n$ we get
\begin{align} \label{EqA.1}
A(j)&:=\Big|\prod_{k=j}^n  \sin\big( \sum_{m=k+1}^{n} \Im \langle f_k | f_m\rangle+ 
\Im \langle f_k | f\rangle \big)\Big|\\ \nonumber
&\leq  e^{-\gamma  t_j}  |\lambda_j| \alpha  \prod_{k = j+1}^{n}(1 + \beta  \lambda_k^2). 
\end{align}
We use the convention that $\sum_{m=n+1}^{n}(*)=0$ and $\prod_{m=n+1}^{n}(*)=1$.
\end{Lemma}
\begin{proof}
First we remark that for $x, y\in \R$
\begin{gather*}
|\sin(x + y)| \leq  |\sin(x)| + |\sin(y)|,
\quad |\sin(x)| \leq  |x|,\\ |\sin(x)| \leq  1
\end{gather*}
We proceed by induction for $j$. We assume $j\leq n-1$ and that $A(j+1),\ldots,A(n)$ obey Estimate \eqref{EqA.1}.
Since $0\leq A(i)\leq A(i+1)\leq 1$ for all $i=1,\ldots,n$ we get
\begin{align}\label{EqA.2}
A(j)
&\leq A(j+1)
\cdot  \Big(  \sum_{m = j+1}^{n} \big|\sin(\Im \langle f_j | f_m\rangle)\big|\\ \nonumber
&\phantom{\leq A(j+1) \cdot  \Big(  \sum_{m = j+1}^{n}  \big|\sin(\Im \langle } 
+\big|\sin(\Im \langle f_j | f\rangle)\big| \Big)\\ \nonumber
&\leq 
   \sum_{m = j+1}^{n} \big|\sin(\Im \langle f_j | f_m\rangle)\big| A(m)
+\big|\sin(\Im \langle f_j | f\rangle)\big| \\ \nonumber
&\leq   \sum_{m = j+1}^{n} \Big( \beta  |\lambda_m \lambda_j|\cdot e^{-(t_j - t_m)  \gamma}\Big)\\ \nonumber
&\qquad 
\Big( e^{-\gamma  t_m} |\lambda_m| \alpha 
    \prod_{k = m+1}^{n}(1 + \beta  \lambda_k^2)\Big)
+ \alpha  |\lambda_j| e^{-t_j  \gamma}.
\end{align}
Since the right-hand side (r.h.s) of Equation \eqref{EqA.2} is less
than the (r.h.s) of Equation  \eqref{EqA.1}, we obtain Lemma \ref{LemA1}.
\end{proof}


\begin{thebibliography}{1}
\bibitem{Arai1981a}
A.~Arai.
\newblock On a Model of a Harmonic Oscillator coupled to a 
Quantized,
Massless, Scalar Field I.
\newblock {\em J.~Math.~Phys.}, 22:2539--2548, 1981.
%
\bibitem{Arai1981c}
A.~Arai.
\newblock On a Model of a Harmonic Oscillator coupled to a 
Quantized,
Massless, Scalar Field II.
\newblock {\em J.~Math.~Phys.}, 22:2549--2552, 1981.
%
\bibitem{ArakiWoods1963}
H.~Araki, E.~Woods.
\newblock Representations of the Canonical Commutation Relations
          describing a non-relativistic infinite free Bose gas.
\newblock {\em J.~Math. Phys.}, 4:637--662, 1963.
%
\bibitem{Attal2006}
S.~Attal, A.~Joye, C.~A.~Pillet.
\newblock Open Quantum Systems I, The Hamiltonian approach.
\newblock Lecture Notes in Mathematics, Springer-Verlag, 2006.
%
\bibitem{BachFroehlichSigal2000}
V. Bach, J. Fröhlich, I.~M.~Sigal. 
\newblock Return to Equilibrium.
\newblock{\em J.~Math.~Phys.}, 41:3985--4060, 2000.
%
%
\bibitem{Berezin1966}
F.~A.~Berezin.
\newblock The method of second quantization.
\newblock Academic Press, New York, 1966.
%
\bibitem{BotvichMaassen2009}
H.~Maassen, D.~Botvich.
\newblock A {G}alton-{W}atson estimate for {D}yson series.
\newblock  { \em Ann. Henri Poincar\'e}, 10: 1141--1158, 2009.
%
\bibitem{BratteliRobinson1987}
O.~Bratteli, D.~Robinson.
\newblock Operator Algebras and Quantum Statistical Mechanics 1.
\newblock Text and Monographs in Physics, Springer-Verlag, 1987.
%
\bibitem{BratteliRobinson1996}
O.~Bratteli, D.~Robinson.
\newblock Operator Algebras and Quantum Statistical Mechanics 2.
\newblock Text and Monographs in Physics, Springer-Verlag, 1996.
%
\bibitem{DerezinskiJaksic2001}
J.~Derezinski, V.~Jaksic.
\newblock Spectral theory of Pauli-Fierz operators.
\newblock {\em Journ. Func. Analysis}, 180 : 241-327, 2001.
%
\bibitem{DerezinskiJaksic2003}
J.~Derezinski, V.~Jaksic.
\newblock Return to Equilibrium for {P}auli-{F}ierz Operators.
\newblock {\em Ann. Henri Poincare}, 4 : 739--793, 2003.
%
\bibitem{FidaleoLiverani1999}
F.~Fidaleo, C.~Liverani.
\newblock Ergodic Properties {F}or {A} {Q}uantum {N}on {L}inear {D}ynamics.
\newblock {\em J. Stat. Phys.}, 97 : 957-1009, 1999.
%
%
\bibitem{FroehlichMerkliSigal2004}
J.~Fr{\"{o}}hlich, M.~Merkli, I.~M.Sigal.
\newblock Ionization of Atoms in a Thermal Field.
\newblock {\em Journal of Statistical Physics}, 116: 311--359, 2004.
%
\bibitem{HaagHugenholzWinnink1967}
R.~Haag, N.~Hugenholz, M.~Winnink.
\newblock On the Equilibrium States in Quantum Statistical Mechanics.
\newblock {\em Commun. Math. Phys.}, 5: 215--236, 1967.
%
%
\bibitem{JaksicPillet1996a}
V.~Jak\v{s}i{\'c}, C.~A.~Pillet.
\newblock On a Model for Quantum Friction. {I}{I}: {F}ermi's Golden Rule 
         and Dynamics at Positive Temperature.
\newblock{ \em Commun. Math. Phys.}, 176:619--643, 1996
%
\bibitem{JaksicPillet1996b}
V.~Jak\v{s}i{\'c}, C.~A.~Pillet.
\newblock On a Model for Quantum Friction {I}{I}{I}: Ergodic 
         Properties of the Spin-Boson System.
\newblock{ \em Commun. Math. Phys.}, 178:627--651, 1996
%
\bibitem{Koenenberg2009b}
M.~Könenberg.
\newblock An Infinite Level Atom coupled to a Heat Bath.
\newblock In preparation.
%
\bibitem{Maassen1983}
H.~Maassen.
\newblock Return to Thermal Equilibrium by the Solution of
a Quantum Langevin Equation.
\newblock {\em J. Stat. Phys.}, 34: 239--262, 1984.
%
\bibitem{Merkli2001}
M.~Merkli.
\newblock Positive Commutators in Non-Equilibrium Statistical Quantum Mechanics.
\newblock {\em Commun.~Math.~Phys.}, 223: 327--362, 2001.  
%
\bibitem{ReedSimonII1980}
M.~Reed, B.~Simon.
\newblock Methods of Modern Mathematical Physics: {I}{I}. {F}ourier Analysis 
         and Self-Adjointness
\newblock Academic Press, 1980.
%
\bibitem{Spohn1997}
H.~Spohn.
\newblock Asymptotic completeness for Rayleigh scattering
\newblock {\em J. Math. Phys.}, 38: 2281--2296, 1997.
%
\end{thebibliography}
\end{document}